\documentclass[12pt]{article}

\usepackage[T1]{fontenc}
\usepackage{latexsym,color,amsmath,amssymb,amsthm}
\usepackage{floatflt}
\usepackage{color}
\usepackage{times}
\usepackage{euscript}
\usepackage{epsfig}
\usepackage{graphicx}
\usepackage{mathpazo}
\usepackage[cm]{fullpage}%
\usepackage{microtype}
\usepackage{fixltx2e}

\usepackage{fullpage}
\usepackage{xspace}

\usepackage{ifpdf}

\graphicspath{{figs/}}

\DeclareGraphicsExtensions{.pdf}

\ifpdf
\else
  \DeclareGraphicsExtensions{.eps}
\fi

\newcommand{\si}[1]{#1}
\newcommand{\remove}[1]{}

\def\A{\EuScript{A}}
\def\C{\EuScript{C}}
\def\D{\EuScript{D}}
\def\E{\EuScript{E}}
\def\G{\EuScript{G}}
\def\J{\EuScript{J}}
\def\K{\EuScript{K}}
\def\R{\EuScript{R}}
\def\V{\EuScript{V}}
\def\T{\EuScript{T}}
\def\union{\EuScript{U}}
\def\U{\EuScript{U}}
\def\EE{\mathsf{E}}

\def\X{\mathsf{X}}
\def\AA{\mathsf{A}}
\def\fdepth{\omega}
\def\asize{\nu}

\def\eps{{\varepsilon}}
\def\reals{{\mathbb R}}
\def\ph{{\varphi}}
\def\preals{{\reals_{{\ge}0}}}
\def\etal{\textsl{et~al.}}

\def\bd{{\partial}}

\def\lpt{a}
\def\rpt{b}

\def\RR{\mathrm{RR}}
\def\RRR{\Theta}
\def\CR{\mathrm{CR}}
\def\radii{\theta}
\def\rr{\mathbf{r}}

\newcommand{\pdf}{density\xspace}
\newcommand{\cdf}{cdf\xspace}
\newcommand{\halfedge}{portal\xspace}
\newcommand{\halfedges}{{\halfedge{}}s\xspace}
\def\compl{\psi}

\newtheorem{theorem}{Theorem}[section]
\newtheorem{lemma}[theorem]{Lemma}
\newtheorem{cor}[theorem]{Corollary}

\newcommand{\lemlab}[1]{\label{lem:#1}}
\newcommand{\figlab}[1]{\label{fig:#1}}

\def\figref#1{Figure~\ref{fig:#1}}

\makeatletter
\long\def\@makecaption#1#2{
   \vskip 10pt
   \setbox\@tempboxa\hbox{{\footnotesize {\bf #1.} #2}}
   \ifdim \wd\@tempboxa >\hsize         % IF longer than one line:
       {\footnotesize {\bf #1.} #2\par}% THEN set as ordinary paragraph.
     \else                              %   ELSE  center.
       \hbox to\hsize{\hfil\box\@tempboxa\hfil}
   \fi}
\makeatother

\begin{document}

\title{Union of Random Minkowski Sums and Network Vulnerability
   Analysis%
   \thanks{%
      A preliminary version of this paper appeared in \emph{Proc. 29th Annual Symposium of Computational Geometry, 2013, pp. 177--186.}
      %\cite{aks-urmsn-13}.
      %
      Work by Pankaj Agarwal and Micha Sharir has been supported by
      Grant 2012/229 from the U.S.-Israel Binational Science Foundation.
      Work by Pankaj Agarwal has also been supported by NSF under grants
      CCF-09-40671, CCF-10-12254, and CCF-11-61359, by ARO grants
      W911NF-07-1-0376 and W911NF-08-1-0452, and by an ERDC contract
      W9132V-11-C-0003. %
      Work by Sariel Har-Peled has been supported by NSF under
         grants CCF-0915984 and CCF-1217462. %
      Work by Haim Kaplan has been supported by grant 822/10 from the
      Israel Science Foundation, grant 1161/2011 from the
      German-Israeli Science Foundation, and by the Israeli Centers
      for Research Excellence (I-CORE) program (center no.~4/11). %
      Work by Micha Sharir has also been supported by NSF Grant
      CCF-08-30272, by Grants 338/09 and 892/13 from the Israel
      Science Foundation, by the Israeli Centers for
      Research Excellence (I-CORE) program (center no.~4/11), and by
      the Hermann Minkowski--MINERVA Center for Geometry at Tel Aviv
      University.}%
}

\author{%
   Pankaj K. Agarwal\thanks{%
      Department of Computer Science, Box 90129, Duke University,
      Durham, NC 27708-0129, USA; {\tt pankaj@cs.duke.edu}}%
   \and Sariel Har-Peled\thanks{%
      Department of Computer Science; University of Illinois; 201 N.
      Goodwin Avenue; Urbana, IL, 61801, USA; {\tt
         \si{sariel@illinois.edu}}} \and Haim Kaplan\thanks{%
      School of Computer Science, Tel Aviv University, Tel~Aviv 69978,
      Israel; {\tt haimk@post.tau.ac.il }} \and Micha Sharir\thanks{%
      School of Computer Science, Tel Aviv University, Tel~Aviv 69978,
      Israel; {\tt michas@post.tau.ac.il }} }

% \date{}

\begin{titlepage}
    \maketitle

    \begin{abstract}
        Let $\C=\{C_1,\ldots,C_n\}$ be a set of $n$
    pairwise-disjoint convex sets of constant description complexity,
    and let $\pi$ be a probability density function (\pdf for short) over
    the non-negative reals.  For each $i$, let $K_i$ be the
        Minkowski sum of $C_i$ with a disk of radius $r_i$, where each
        $r_i$ is a random non-negative number drawn independently from
        the distribution determined by $\pi$. We show that the expected
    complexity of the union of $K_1, \ldots, K_n$ is
    $O(n^{1+\eps})$ for any $\eps > 0$; here the constant of
        proportionality depends on $\eps$ and on the description
        complexity of the sets in $\C$, but not on $\pi$.
    If each $C_i$ is a convex polygon with at most $s$ vertices, then
    we show that the expected complexity of the union
    %of $K_1, \ldots, K_n$
    is $O(s^2n\log n)$.

        Our bounds hold in the stronger model in which we are given an arbitrary
        multi-set $\RRR=\{\radii_1,\ldots,\radii_n\}$ of expansion
        radii, each a non-negative real number. We assign them to the
    members of $\C$ by a random
        permutation, where all permutations are equally likely to be
    chosen; the expectations are now with respect to these permutations.

        We also present an application of our results to a
    problem that arises in analyzing the vulnerability of a network to
        a physical attack.
        \remove{ The first bound has an application to the following
           problem that arises in analyzing the vulnerability of a
           network under a physical attack. Let $\G=(V,\E)$ be a
           planar geometric graph where $\E$ is a set of $n$ line
           segments with pairwise-disjoint relative interiors.  Let
           $\ph:\preals \rightarrow [0,1]$ be an \emph{edge failure
              probability function}, where a physical attack at a
           location $x\in\reals^2$ causes an edge $e$ of $\E$ at
           distance $r$ from $x$ to fail with probability $\ph(r)$; we
           assume that $\ph$ is of the form $1-\Pi(x)$, where $\Pi$ is
           a cumulative distribution function (\cdf) on the
           non-negative reals. The goal is to compute the most
           \emph{vulnerable} location for $\G$, i.e., the location of
           the attack that maximizes the expected number of failing
           edges of $\G$.  Using our bound on the complexity of the
           union of random Minkowski sums, we present a near-linear
           Monte-Carlo algorithm for computing a location that is an
           approximately most vulnerable location of attack for $\G$.}

    \end{abstract}
\end{titlepage}

\section{Introduction}
\label{sec:intro}

\paragraph{Union of random Minkowski sums.}
Let $\C=\{C_1,\ldots,C_n\}$ be a set of $n$ pairwise-disjoint convex sets
of constant description complexity, i.e., the boundary of each $C_i$ is defined by a constant
number of algebraic arcs of constant maximum degree.
%Let $\C = \{C_1,C_2,\ldots,C_n\}$ be a collection of $n$
%pairwise-disjoint convex polygons, each with at most $s$ edges, for
%some constant $s$.
Let $D(r)$ denote the disk of radius $r$ centered
at the origin. We consider the setup where we are given a sequence
$\rr=\langle r_1, \ldots, r_n\rangle$ of non-negative numbers, called
\emph{expansion distances} (or \emph{radii}).
We set $K_i =
C_i \oplus D(r_i)$, the Minkowski sum of $C_i$ with $D(r_i)$. The
boundary of $K_i$, denoted by $\bd K_i$,
consists of $O(1)$ algebraic arcs of bounded degree.
If $C_i$ is a convex polygon with $s$ vertices, then $\bd K_i$
is an alternating concatenation of line segments and
circular arcs, where each segment is a parallel shift, by distance
$r_i$, of an edge of $C_i$, and each circular arc is of radius $r_i$
and is centered at a vertex of $C_i$; see Figure~\ref{fig:mink}. We
refer to the endpoints of the arcs of $\bd K_i$
as the \emph{vertices} of $K_i$.
Let $\K = \{K_1, \ldots, K_n\}$, and let $\union =
\union(\K)= \bigcup_{i=1}^n K_i$. The \emph{combinatorial complexity}
of $\union$, denoted by $\compl(\C,\rr)$, is defined to be the number of
vertices of $\union$, each of which is either a vertex of some $K_i$ or
an intersection point of
the boundaries of a pair of $K_i$'s, lying on $\bd \union$.  We do not
make any assumptions on the shape and location of the sets in
$\C$, except for requiring them to be pairwise disjoint.
% and also assuming them to be in general position (see below for details); the
%latter assumption is made only to simplify the analysis, and is not
%essential, so the bound that we will derive thus holds for any family
%of $n$ pairwise-disjoint convex $s$-\si{gons}.

\begin{figure}[\si{htbp}]
    \centering%
%    \scalebox{0.8}
    {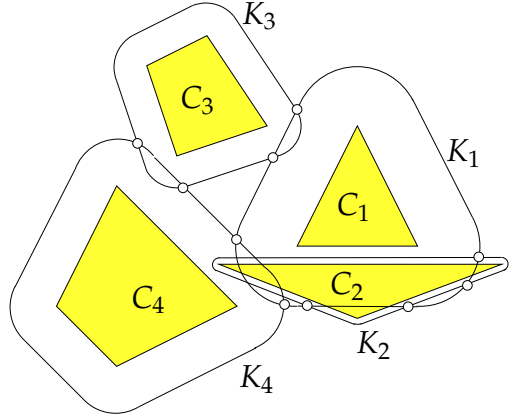}%
    \caption{Pairwise-disjoint convex polygons and their Minkowski
       sums with disks of different radii. The vertices of the union
       of these sums are highlighted.} \label{fig:mink}
\end{figure}

%The case where $\C$ consists of polygons is a simple special case of
%the more general case where $\C$ is a collection of $n$ arbitrary
%pairwise-disjoint convex sets of constant description complexity (see
%below for the precise definition). We also address this general setup
%in this paper.

Our goal is to obtain an upper bound on the \emph{expected}
combinatorial complexity of $\union$, under a suitable probabilistic
model for choosing the expansion radii $\rr$ of the members of
$\C$---see below for the precise models that we will use.
%%%%%%%%%%

\paragraph{Network vulnerability analysis.}
Our motivation for studying the above problems comes from the problem
of analyzing the vulnerability of a network to a physical attack
(e.g., electromagnetic pulse (EMP) attacks, military bombing, or
natural disasters~\cite{EMP}), as studied in~\cite{AEG*}.
Specifically, let $\G=(V,\E)$ be a planar graph embedded in the plane,
where $V$ is a set of points in the plane and $\E=\{e_1, \ldots,
e_n\}$ is a set of $n$ segments (often called \emph{links}) with
pairwise-disjoint relative interiors, whose endpoints are points of
$V$.  For a point $q\in \reals^2$, let $d(q,e)= \min_{p\in e}\|q-p\|$
denote the (minimum) distance between $q$ and $e$. Let $\ph: \preals
\rightarrow [0,1]$ denote the \emph{edge failure probability
function}, so that the probability of an edge $e$ to be damaged by a
physical attack at a location $q$ is $\ph(d(q,e))$. In this model, the
failure probability only depends on the distance of the point of
attack from $e$. We assume that $1-\ph$ is a cumulative distribution
function (\cdf), or, equivalently, that $\varphi(0)=1$,
$\varphi(\infty)=0$, and $\varphi$ is monotonically decreasing.  A
typical example is $\ph(x) = \max\{1-x,0\}$, where the \cdf is the
uniform distribution on $[0,1]$.

\begin{figure}[\si{htb}]
     \centering
     \begin{tabular}{\si{ccc}}
         \includegraphics[width=2.75in]{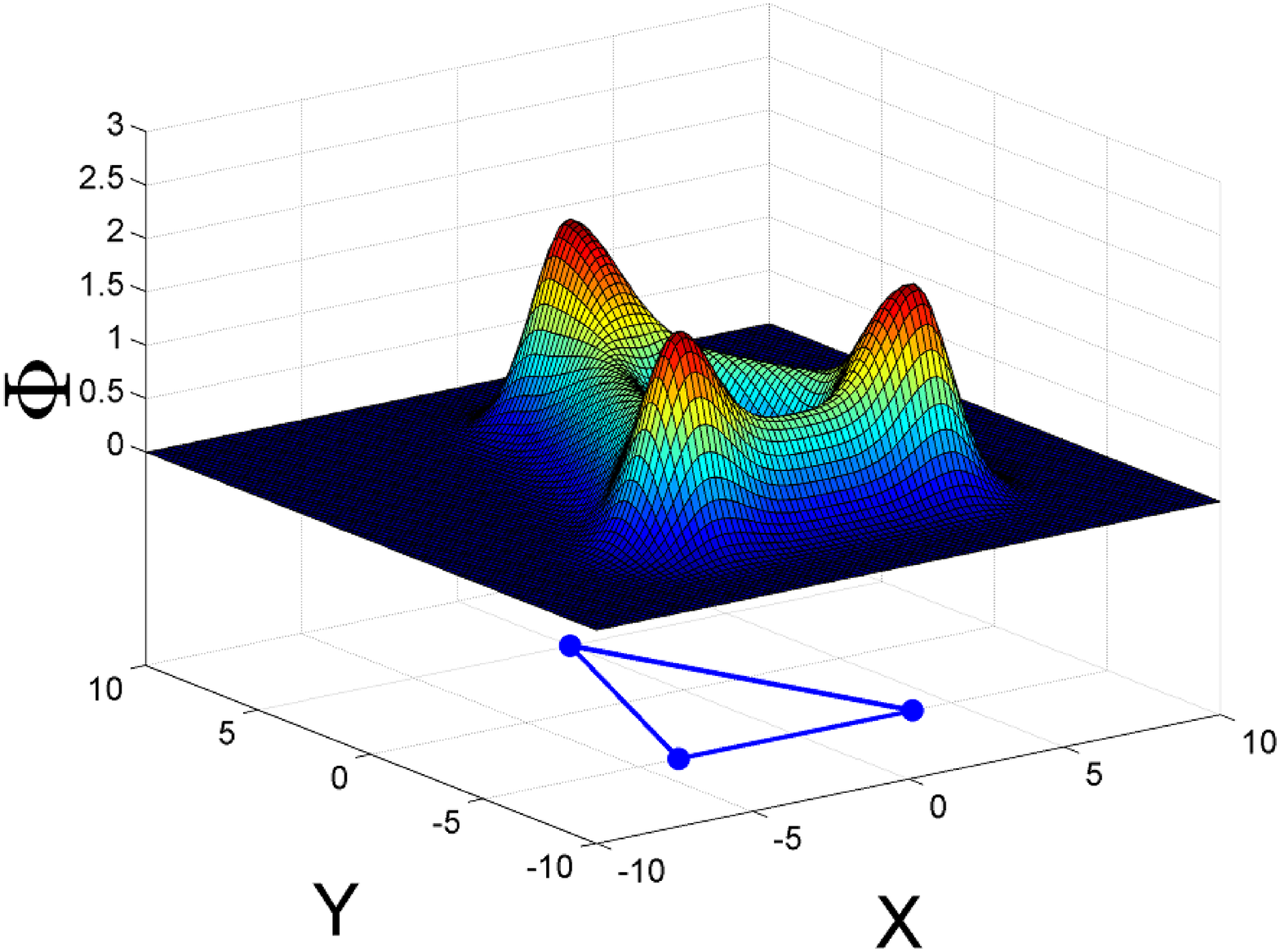}&\hspace{5mm}&
         \includegraphics[width=2.75in]{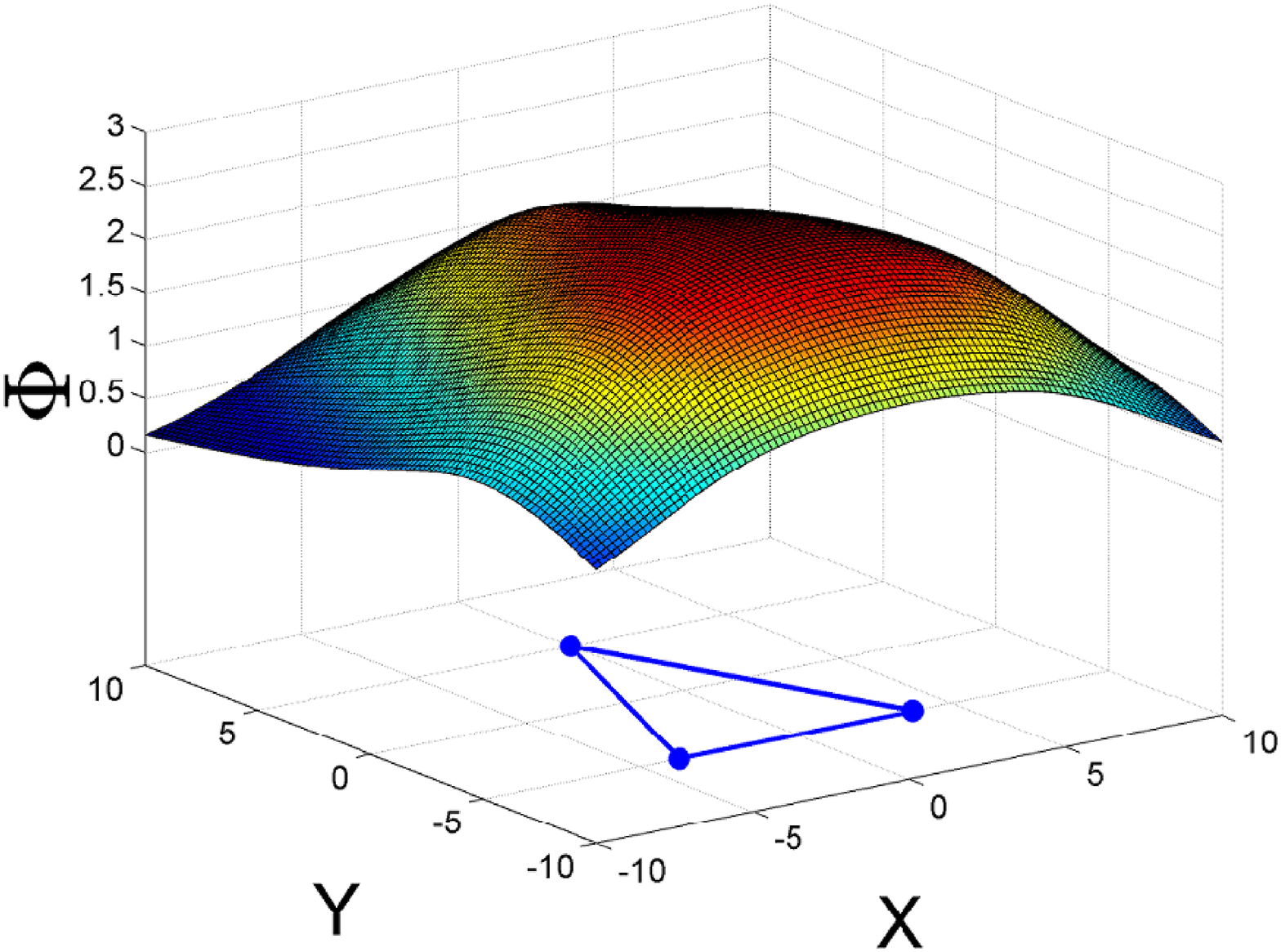}\\
         \small (i)&&\small(ii)
     \end{tabular}
     \caption{Expected damage for a triangle network and Gaussian
        probability distribution function with (i) small variance,
        (ii) large variance.}
     \label{fig:nva}
\end{figure}

For each $e_i \in \E$, let $f_i(q)=\ph(d(q,e_i))$.  The function
$\Phi(q,\E)=\sum_{i=1}^n f_i(q)$ gives the expected number of links of
$\E$ damaged by a physical attack at a location $q$; see
Figures~\ref{fig:nva} and \ref{fig:nva2}. Set
$$\Phi(\E)=\max_{q \in \reals^2} \Phi(q,\E).$$
Our ideal goal is to compute $\Phi(\E)$ and a location
% $q^*=\arg\max_{q\in\reals^2} \Phi(q,\E)$ that maximizes this
% expression.
$q^*$ such that $\Phi(q^*,\E)=\Phi(\E)$.  We refer to such a point
$q^*$ as a \emph{most vulnerable} location for $\G$.  As evident from
Figure~\ref{fig:nva2}, the function $\Phi$ can be quite complex, and
it is generally hard to compute $\Phi(\E)$ exactly, so we focus on
computing it approximately. More precisely, given an error parameter
$\delta > 0$, we seek a point $\tilde{q} \in \reals^2$ for which
$\Phi(\tilde{q},\E) \ge (1-\delta) \Phi(\E)$ (a so-called
\emph{approximately most vulnerable} location).
Agarwal~\etal~\cite{AEG*} proposed a Monte Carlo algorithm for this
task.  As it turns out, the problem can be reduced to the problem of
estimating the maximal depth in an arrangement of random Minkowski
sums of the form considered above, under the \pdf model, and its
performance then depends on the expected complexity of $\union(\K)$.
Here $\K$ is a collection of Minkowski sums of the form $e_i
\oplus D(r_i)$, for a sample of edges $e_i \in \E$ and for suitable
random choices of the $r_i$'s, from the distribution $1-\ph$.  We
adapt and simplify the algorithm in~\cite{AEG*} and prove a better
bound on its performance by using the sharp (near-linear) bound on the
complexity of $\union(\K)$ that we derive in this paper; see below and
Section~\ref{sec:nva} for details.

\begin{figure}[\si{htb}]
    \centering
    \includegraphics[width=3.25in]{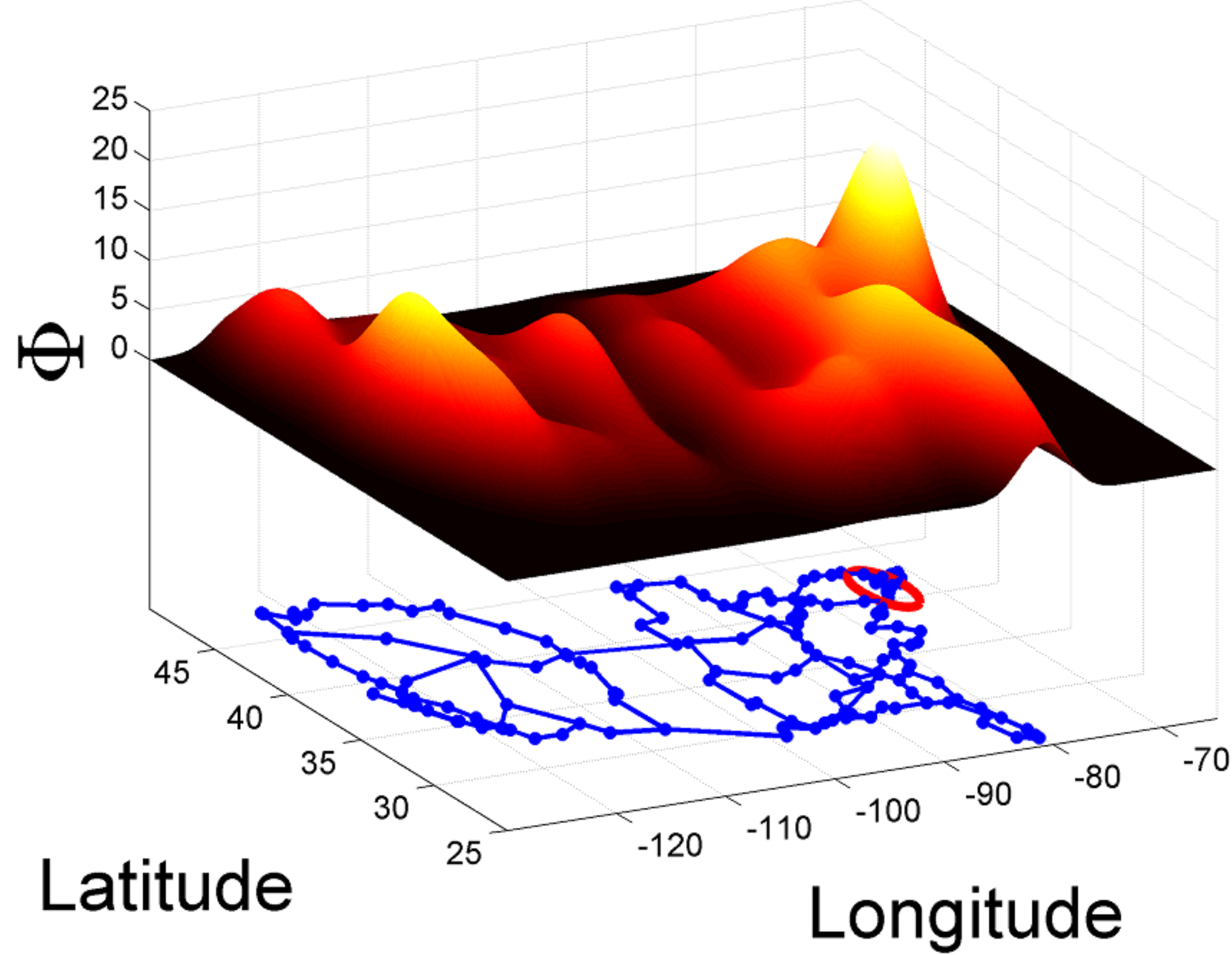}
    \caption{Expected damage for a complex fiber network. This figure
       is taken from~\cite{AEG*}.} \label{fig:nva2}
\end{figure}

\paragraph{Related work.}
\textbf{(i) \textit{Union of geometric objects.}}  There is
extensive work on bounding the complexity of the union of a set of
geometric objects, especially in $\reals^2$ and $\reals^3$, and
optimal or near-optimal bounds have been obtained for many interesting
cases. We refer the reader to the survey paper by
Agarwal~\etal~\cite{APS} for a comprehensive summary of most of the
known results on this topic. For a set of $n$ planar objects, each of
\emph{constant description complexity}, the complexity of their union
can be $\Theta(n^2)$ in the worst case, but many linear or near-linear
bounds are known for special restricted cases. For example, a fairly
old result of Kedem~\etal~\cite{KLPS} asserts that the union of a set
of pseudo-disks in $\reals^2$
%(simply connected regions, where any pair of object boundaries intersect in at most two points)
has linear complexity.  It is also shown in \cite{KLPS} that the Minkowski sums
of a set of pairwise-disjoint planar convex objects with a fixed
common convex set is a family of pseudo-disks.
%(under a suitable assumption of general position).
Hence, in our setting, if all the
$r_i$'s were equal, the result of \cite{KLPS} would then imply that
the complexity of $\union(\K)$ is $O(n)$. On the other hand, an adversial
choice of the $r_i$'s may result in a union $\union$ with
$\Theta(n^2)$ complexity; see Figure~\ref{race2}.

\begin{figure}[\si{htb}]
    \centering %
    \scalebox{0.8}
    {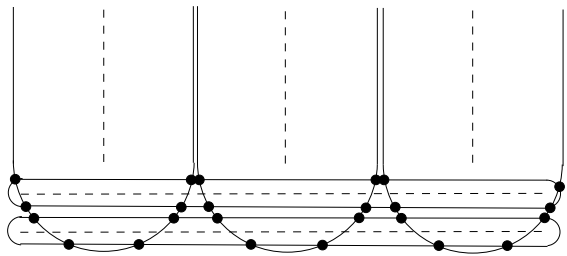}
    \caption{A bad choice of expansion distances may cause $\union$ to
       have quadratic complexity.}%
    \label{race2}
\end{figure}

\noindent\textbf{(ii) \textit{Network vulnerability analysis.}}
Most of the early work on network vulnerability analysis considered a small
number of isolated, independent failures; see, e.g.,
\cite{Bhandari99,OuBook} and the references therein.  Since physical
networks rely heavily on their physical infrastructure, they are
vulnerable to physical attacks such as electromagnetic pulse (EMP)
attacks as well as natural disasters~\cite{EMP,W-top}, not to mention
military bombing and other similar kinds of attack. This has led to
recent work on analyzing the vulnerability of a network under
geographically correlated failures due to a physical attack at a
single location~\cite{AEG-MCOM,AEG*,NM09TR,G42,W-top}. Most papers on
this topic have studied a deterministic model for the damage caused by
such an attack, which assumes that a physical attack at a location $x$
causes the failure of all links that intersect some simple geometric
region (e.g., a vertical segment of unit length, a unit square, or a
unit disk) centered at $x$. The impact of an attack is measured in
terms of its effect on the connectivity of the network, (e.g., how many
links fail, how many pairs of nodes get disconnected, etc.), and the
goal is to find the location of attack that causes the maximum damage
to the network. In the simpler model studied in \cite{AEG*} and in the
present paper, the damage is measured by the number of failed
links. This is a problem that both attackers and planners of such
networks would like to solve. The former for obvious reasons, and the
latter for identifying the most vulnerable portions of the network, in
order to protect them better.

In practice, though, it is hard to be certain in advance whether a
link will fail by a nearby physical attack. To address this situation,
Agarwal \etal~\cite{AEG*} introduced the simple probabilistic
framework for modeling the vulnerability of a network under a physical
attack, as described above.
% They studied various measures of the impact of an attack on the
% connectivity of the network.
One of the problems that they studied is to compute the largest
expected number of links damaged by a physical attack. They described
an approximation algorithm for this problem whose expected running
time is quadratic in the worst case.
A major motivation for the present study is to improve the efficiency
of this algorithm and to somewhat simplify it at the same time.

Finally, we note that the study in this paper has potential
applications in other contexts, where one wishes to analyze the
combinatorial and topological structure of the Minkowski sums (or
rather convolutions) of a set of geometric objects (or a function over
the ambient space) with some \emph{kernel function} (most notably a
Gaussian kernel), or to perform certain computations on the resulting
configuration. Problems of this kind arise in many applications,
including statistical learning, computer vision, robotics, and
computational biology; see, e.g.,~\cite{EFR,LP} and references
therein.

\paragraph{Our models.}
We consider two probabilistic models for choosing the sequence $\rr=\langle r_1,\ldots,r_n\rangle$ of expansion distances:
%The proofs of the bounds that are stated below are carried out in the
%permutation model, but since, as we will shortly argue, this model is
%stronger than the \pdf model, the bounds also hold under this latter
%model.  For the convenience of the reader we recall the two models.

\medskip
\noindent\textbf{\textit{The \pdf model.}}
We are given an arbitrary \pdf (or a probability mass function) $\pi$ over the non-negative reals; for
each $1 \le i \le n$, we take $r_i$ to be a  random
value drawn independently from the distribution determined by $\pi$.

\medskip
\noindent\textbf{\textit{The permutation model.}}
We are given a multi-set $\RRR=\{\radii_1,\ldots,\radii_n\}$ of $n$ arbitrary
non-negative real numbers. We draw a random permutation $\sigma$ on
$[1:n]$, where all permutations are equally likely to be
chosen, and assign $r_i:=\radii_{\sigma(i)}$ to $C_i$ for each
$i=1,\ldots,n$.

\medskip

Our goal is to prove sharp bounds on on the expected complexity of
the union $\union(\K)$ under these two models. More precisely, for
the density model, let $\compl(\C,\pi)$ denote the expected value of
$\compl(\C,\rr)$, where the expectation is taken over the random
choices of $\rr$, made from $\pi$ in the manner specified above. Set
$\compl(\C)=\max \compl(\C,\pi)$, where the maximum is taken over
all probability density (mass) functions. For the permutation model,
in an analogous manner, we let $\compl(\C,\RRR)$ denote the expected
value of $\compl(\C,\rr)$, where the expectation is taken over the
choices of $\rr$, obtained by randomly shuffling the members of
$\RRR$. Then, with a slight overloading of the notation, we define
$\compl(\C) = \max \compl(\C,\RRR)$, where the maximum is over all
possible choices of the multi-set $\RRR$. We wish to obtain an upper
bound on $\compl(\C)$ under both models.

We note that the permutation model is more general than the \pdf
model, in the sense that an upper bound on $\compl(\C)$ under the
permutation model immediately implies the same bound on $\compl(\C)$
under the \pdf model. Indeed,
% let $\kappa_{\mathrm{perm}}(n)$
%denote an upper bound for $\compl(\C)$ under the permutation model.
consider some given \pdf $\pi$, out of whose distribution the
distances $r_i$ are to be sampled (in the \pdf model). Interpret
such a random sample as a 2-stage process, where we first sample
from the distribution of $\pi$ a multi-set $\RRR$ of $n$ such
distances, in the standard manner of independent repeated draws, and
then assign the elements of $\RRR$ to the sets $C_i$ using a random
permutation (it is easily checked that this reinterpretation does
not change the probability space). Let $\rr$ be the resulting
sequence of expansion radii for the members of $\C$.
%The resulting probability space is identical under both interpretations.
Using the new interpretation, the  expectation of $\compl(\C,\rr)$
(under the \pdf model), conditioned on the fixed $\RRR$, is at most
$\compl(\C)$ under the permutation model.
%$\kappa_{\mathrm{perm}}(n)$.
Since this bound holds for every $\RRR$, the unconditional expected
value of $\compl(\C,\rr)$ (under the \pdf model) is also at most
$\compl(\C)$ under the permutation model.
% $\kappa_{\mathrm{perm}}(n)$.
Since this holds for every $\pdf$, the
claim follows.

We do not know whether the opposite inequality also holds. A natural
reduction from the permutation model to the \pdf model would be to
take the input set $\RRR$ of the $n$ expansion distances and regard
it as a discrete mass distribution (where each of its members can be
picked with probability $1/n$). But then, since the draws made in
the \pdf model are independent, most of the draws will not be
permutations of $\RRR$, so this approach will not turn $\compl(\C)$
under the \pdf model  into an upper bound for $\compl(\C)$ under the
permutation model.
%
% It is clear that the two models are not identical.
% For example, consider the \pdf where
% radius $1$ and $2$ are picked with equal probability. The probability
% that the set of radii generated, for $n$ objects, would have exactly
% half the values $1$ and exactly half the values $2$ is $\binom{n}{n/2}
% = \Theta (1/\sqrt{n} )$. Naturally, if we set $\RRR$ be this specific
% multi-set, then in the permutation model we will get this set of
% radii.
%This subtle difference is a minor technical issue, and seem to have no impact on the bounds we get.
%We leave the question of whether there exist a collection $\C$ and a
%\pdf $\pi$ for which $\compl(\C)$ under the pdf model is strictly
%smaller than   $\compl(\C)$ under the permutation model for future
%research.
%

\paragraph{Our results.}
The main results of this paper are near-linear upper bounds on $\compl(\C)$
under the two models discussed above. Since the permutation model is more
general, in the sense made above, we state, and prove, our results in this model.
We obtain (in Section~\ref{sec:blobs}) the following bound for the general case.

%%%%%%%%%%%%%%%%%%%%%%%%%%%%%%%%%%%%%%%%%%
\begin{theorem} \label{thm:sets}
    Let $\C=\{C_1,\ldots,C_n\}$ be a set
    of $n$ pairwise-disjoint convex sets of constant description
    complexity in $\reals^2$.  Then the value of $\compl(\C)$ under the
    permutation model, for any multi-set $\RRR$ of expansion radii,
    is $O(n^{1+\eps})$, for any $\eps > 0$;
    the constant of proportionality depends on $\eps$ and the description
    complexity of the members of $\C$, but not on $\RRR$.
\end{theorem}
%%%%%%%%%%%%%%%%%%%%%%%%%%%%%%%%%%%%%%%%%%

If $\C$ is a collection of convex polygons, we obtain a slightly improved bound by
using a different, more geometric argument.

%%%%%%%%%%%%%%%%%%%%%%%%%%%%%%%%%%%%%%%%%%
\begin{theorem} \label{thm:main}
    Let $\C=\{C_1,\ldots,C_n\}$ be a set of $n$ pairwise-disjoint convex
    polygons in $\reals^2$, where each $C_i$ has at most $s$ vertices.
    Then the maximum value of $\compl(\C)$ under the permutation model,
    for any multi-set $\RRR$ of expansion radii, is $O(s^2n \log n)$.
\end{theorem}
%%%%%%%%%%%%%%%%%%%%%%%%%%%%%%%%%%%%%%%%%%

For simplicity, we first prove Theorem~\ref{thm:main} in
Section~\ref{sec:segs} for the special case where $\C$ is a set of $n$
segments with pairwise-disjoint relative interiors.
%say, $e_1, \ldots, e_n$, in which case
%each $K_i = e_i \oplus D(r_i)$ is a \emph{racetrack}, bounded by two
%segments that are parallel shifts, by distance $r_i$, of $e_i$, and by
%two semicircles, of radius $r_i$ each, centered at the endpoints of
%$e_i$; e.g., see Figure~\ref{race2}.
Then we extend the proof, in a
reasonably straightforward manner, to polygons in
Section~\ref{sec:poly}. The version involving segments admits a
somewhat cleaner proof, and is sufficient for the
application to network vulnerability analysis.

Using the Clarkson-Shor argument~\cite{CS}, we also obtain the
following corollary, which will be needed for the analysis in Section
\ref{sec:nva}.
%%%%%%%%%%%%%%%%%%%%%%%%%%%%%%%%%%%%%%%%%%
\begin{cor}%
    \label{cor:klevel}%
    Let $\C=\{C_1,\ldots,C_n\}$ be a set of $n$ pairwise-disjoint
    convex set of constant description compelxity.
    Let $r_1,\ldots,r_n$ be the random expansion distances, obtained under
    the permutation model, for any multi-set $\RRR$ of expansion radii,
    that are assigned to $C_1,\ldots,C_n$,
    respectively, and set $\K = \{C_i \oplus D(r_i) \mid 1\ \le i \le
    n\}$.  Then, for any $1\le k\le n$, the expected number of
    vertices in the arrangement $\A(\K)$ whose depth is at most $k$ is
    $O(n^{1+\eps}k^{1-\eps})$, for any $\eps > 0$;
    the constant of proportionality depends on $\eps$ and the description
    complexity of the members of $\C$ but not on $\RRR$.
    If each $C_i$ is a convex polygon with at most $s$ vertices, then
    the bound improves to $O(s^2nk\log(n/k))$.
\end{cor}
%%%%%%%%%%%%%%%%%%%%%%%%%%%%%%%%%%%%%%%%%%

Using Theorem~\ref{thm:main} and Corollary
\ref{cor:klevel}, we present (in Section~\ref{sec:nva}) an efficient
Monte-Carlo $\delta$-approximation algorithm for computing an
approximately most vulnerable location for a network, as defined
earlier. Our algorithm is a somewhat simpler, and considerably more efficient,
variant of the algorithm
proposed by Agarwal~\etal~\cite{AEG*}, and the general approach is
similar to the approximation algorithms presented
in~\cite{AES,AHSSW,AHP} for computing the depth in an arrangement of a
set of objects. Specifically, we establish the following result.

%%%%%%%%%%%%%%%%%%%%%%%%%%%%%%%%%%%%%%%%%%
\begin{theorem} \label{thm:nva}
    Given a set $\E$ of $n$ segments in
    $\reals^2$ with pairwise-disjoint relative interiors, an
    edge-failure-probability function $\ph$ such that $1-\ph$ is a
    \cdf, and a constant $0< \delta < 1$, one can compute, in
    $O(\delta^{-4}n\log^3 n)$ time, a location $\tilde q \in\reals^2$,
    such that $\Phi(\tilde{q},\E) \ge (1-\delta)\Phi(\E)$ with
    probability at least $1-1/n^c$, for arbitrarily large $c$;
    the constant of proportionality in the running-time bound depends
    on $c$.
\end{theorem}
%%%%%%%%%%%%%%%%%%%%%%%%%%%%%%%%%%%%%%%%%%%%%%%%%%%%%%%%%%%%
\section{The Case of Convex Sets}
\label{sec:blobs}
%%%%%%%%%%%%%%%%%%%%%%%%%%%%%%%%%%%%%%%%%%%%%%%%%%%%%%%%%%%%

In this section we prove Theorem~\ref{thm:sets}.
We have a collection $\C=\{C_1,\ldots,C_n\}$ of $n$
pairwise-disjoint compact convex sets in the plane, each of constant
description complexity.
Let $\RRR$ be a multi-set of $n$ non-negative real numbers
$0 \le \radii_1\le \radii_2\le \cdots \le \radii_n$.  We choose a random
permutation $\sigma$ of $[1:n]$, where all permutations are equally
likely to be chosen, put $r_i = \radii_{\sigma(i)}$ for $i=1,\ldots,n$,
and form the Minkowski sums $K_i = C_i \oplus D(r_i)$, for $i=1,\ldots,n$. We put
$\K=\{K_1,\ldots,K_n\}$. We prove a near-linear upper bound on the expected complexity of
$\U(\K)$, as follows.

Fix a parameter $t$, whose value will be determined later, and put
$\rho=\radii_{n-t}$,  the $(t+1)$-st largest distance in $\RRR$. Put
    $\C^+  = \{ C_i\in\C \mid \sigma(i) > n-t \}$ and
    $\C^-  = \{ C_i\in\C \mid \sigma(i) \le n-t \}$.
That is, $\C^+$ is the set of the $t$ members of $\C$ that were
assigned the $t$ largest distances in $\RRR$, and $\C^-$ is the
complementary subset.

By construction, $\C^+$ is a random subset of $\C$ of size $t$ (where
all $t$-element subsets of $\C$ are equally likely to arise as
$\C^+$). Moreover, conditioned on the choice of $\C^+$, the set $\C^-$
is fixed, and the subset $\RRR^-$ of the $n-t$ distances in $\RRR$
that are assigned to them is also fixed. Furthermore, the permutation
that assigns the elements of $\RRR^-$ to the sets in $\C^-$ is a
random permutation.

For each $C_i\in\C^+$, put $K_i^* = C_i\oplus D(\rho)$. Put $(\K^*)^+
= \{ K_i^* \mid C_i\in\C^+\}$, and let $\U^*$ denote the union of
$(\K^*)^+$. Note that $\U^*\subseteq\U$, because $K_i^*\subseteq K_i$
for each $C_i\in\C^+$. Since the $C_i$'s are pairwise-disjoint and we
now add the same disk to each of them, $(\K^*)^+$ is a collection of
pseudo-disks \cite{KLPS}, and therefore $\U^*$ has $O(t)$ complexity.

Let $\V$ denote the \emph{vertical decomposition} of the complement of
$\U^*$; it consists of $O(t)$ \emph{pseudo-trapezoids}, each defined
by at most four elements of $(\K^*)^+$ (that is, of $\C^+$).  See
\cite{SA} for more details concerning vertical decompositions.  For
short, we refer to these pseudo-trapezoids as \emph{trapezoids}.

In a similar manner, for each $C_i\in\C^-$, define $K_i^* = C_i \oplus
D(\rho)$; note that $K_i\subseteq K_i^*$ for each such $i$.  Put
$(\K^*)^- = \{ K_i^* \mid C_i\in\C^-\}$. Since $\C^+$ is a random
sample of $\C$ of size $t$, the following lemma follows from a standard
random-sampling argument; see~\cite[Section~4.6]{Mat} for a proof.

\begin{lemma}
    With probability $1-O\left(\frac{1}{n^{c-4}}\right)$, every (open)
    trapezoid $\tau$ of $\V$ intersects at most $k:= \frac{cn}{t}\ln
    n$ of the sets of $(\K^*)^-$, for sufficiently large $c>4$.
\end{lemma}
\remove{
\begin{proof}
    Fix a trapezoid $\tau$ among the set $\T$ of $O(n^4)$ possible
    trapezoids that can appear in $\V$ (over all possible choices of $\C^+$).
    Each such trapezoid is defined
    by at most four sets in $\C$ and by $\rho$ (recall that $\rho$ is
    not a random variable; it depends only on the input set $\RRR$ and
    the threshold index $t$).
    Suppose $\tau$ intersects more than $k$ of the sets of $\K^*=\{
    K_i^* \mid C_i\in\C\}$.

    For $\tau$ to appear in $\V$, none of these $k$ sets should be
    sampled in $\C^+$. Since $\C^+$ is a random sample of $\C$ of size
    $t$, the probability of this to happen is at most
    \begin{align*}
        \frac{\binom{n-k}{t}}{\binom{n}{t}} & =
        \frac{(n-k)(n-k-1)\cdots (n-k-t+1)}{n(n-1)\cdots (n-t+1)} \\
        & < \left(\frac{n-k}{n}\right)^t = \left(1 -
            \frac{k}{n}\right)^t < e^{-kt/n} = \frac{1}{n^c} .
    \end{align*}
    Taking the union bound over all $O(n^4)$ trapezoids in $\T$ we get
    that, with probability $1-O\left(\frac{1}{n^{c-4}}\right)$, each
    trapezoid in $\V$ is crossed by at most $\frac{cn}{t}\ln n$ sets
    of $(\K^*)^-$.
\end{proof}
}

For each trapezoid $\tau$ of $\V$, let $\C_\tau^-$ denote the
collection of the sets $C_i\in\C^-$ for which $K_i^*$ crosses $\tau$.
We form the union $\U_\tau^-$ of the ``real'' (and smaller)
corresponding sets $K_i$, for $C_i\in\C_\tau^-$, and clip it to within
$\tau$ (clearly, no other set $C_i\in\C^-$ can have its real expansion
$K_i$ meet $\tau$). Finally, we take all the ``larger'' sets
$C_i\in\C^+$, and form the union of $\U_\tau^-$ with the corresponding
``real'' $K_i$'s, again clipping it to within $\tau$.  The overall
union $\U$ is the union of $\U^*$ and of all the modified unions
$\U_\tau^-$, for $\tau\in\V$.

This divide-and-conquer process leads to the following recursive
estimation of the expected complexity of $\U$.  For $m\le n$, let
$\C'$ be any subset of $m$ sets of $\C$, and let $\RRR'$ be any subset
of $m$ elements of $\RRR$, which we enumerate, with some abuse of
notation, as $\theta_1,\ldots,\theta_m$. Let $T(\C',\RRR')$ denote the expected
complexity of the union of the expanded regions $C_i \oplus
D(\theta_{\sigma'(i)})$, for $C_i\in\C'$, where the expectation is over the
random shuffling permutation $\sigma'$ (on $(1,\ldots,m)$). Let $T(m)$
denote the maximum value of $T(\C',\RRR')$, over all subsets $\C'$ and
$\RRR'$ of size $m$ each, as just defined.

Let us first condition the analysis on a fixed choice of $\C^+$.  This
determines $\U^*$ and $\V$ uniquely. Hence we have a fixed set of
trapezoids, and for each trapezoid $\tau$ we have a fixed set
$\C_\tau^-$ of $k_\tau = |\C_\tau^-|$ sets, whose expansions by $\rho$
meet $\tau$. The set $\RRR_\tau^-$ of distances assigned to these sets
is not fixed, but it is a \emph{random subset} of
$\{\theta_1,\ldots,\theta_{n-t}\}$ of size $k_\tau \le \frac{cn}{t}\ln n$, where
$k_\tau$ depends only on $\tau$.  Moreover, the assignment (under the
original random permutation $\sigma$) of these distances to the sets
in $\C_\tau^-$ is a \emph{random} permutation. Hence, conditioning
further on the choice of $\RRR_\tau^-$, the expected complexity of
$\U_\tau^-$, before its modification by the expansions of the larger
sets of $\C^+$, and ignoring its clipping to within $\tau$, is
$$
T(\C_\tau^-,\RRR_\tau^-) \le T(k_\tau) \le T\left( \frac{cn}{t}\ln
    n\right) .
$$
Hence the last expression also bounds the unconditional expected
complexity of the unmodified and unclipped $\U_\tau^-$ (albeit still conditioned on
the choice of $\C^+$).  Summing this over all $O(t)$ trapezoids $\tau$
of $\V$, we get a bound of at most
$$
at T\left( \frac{cn}{t}\ln n\right) ,
$$
for a suitable absolute constant $a$. Since this bound holds for any
choice of $\C^+$, it also bounds the unconditional expected value of
the sum of the complexities of the unmodified unions $\U_\tau^-$. To
this we need to add the complexity of $\U^*$, which is $O(t)$, the
number of intersections between the boundaries of the unions $\U_\tau$
with the respective trapezoid boundaries, and the
number of intersections between the boundaries of the $t$ larger
expansions and the boundaries of all the expansions, that appear on
the boundary of $\U$. The last two quantities are clearly both
at most $O(nt)$ (the constant
in this latter expression depends on the description complexity of the
sets in $\C$). Altogether, we obtain the following recurrence (reusing
the constant $a$ for simplicity).
$$
T(n) \le at T\left( \frac{cn}{t}\ln n\right) + ant ,
$$
which holds when $n$ is sufficiently large. When $n$ is small, we use
the trivial bound $T(n) = O(n^2)$.

With appropriate choice of parameters, the solution of this recurrence
is $T(n) \le An^{1+\eps}$, for any $\eps>0$, where $A$ depends on
$\eps$ and on the other constants appearing in the recurrence. For
this, one needs to choose $t\gg (\log n)^{(1+\eps)/\eps}$, and then
choose $A$ sufficiently large so that the additive term is
significantly subsumed by the other terms, and so that the quadratic
bound for small values of $n$ is also similarly subsumed.  Leaving out
the remaining routine details, we have thus established the bound
asserted in the theorem.  \hfill $\Box$

%%%%%%%%%%%%%%%%%%%%%%%%%%%%%%%%%%%%%%%%%%

\section{The Case of Segments}
\label{sec:segs}

Let $\E=\{e_1,\ldots,e_{n}\}$ be a collection of $n$
line segments in the plane with pairwise-disjoint relative interiors, and
as in Section~\ref{sec:blobs}, let $\RRR$ be a multi-set of $n$ non-negative real numbers
$0 \le \radii_1 \le \radii_2 \le \cdots \le \radii_n$.
For simplicity, we assume that the segments in $\E$ are in general
position, i.e., no segment is vertical, no two of them share an endpoint,
and no two are parallel. Moreover, we
assume that the expansion distances in $\RRR$ are positive and in ``general
position'' with respect to $\E$, so as to ensure that, no matter which
permutation we draw, the racetracks of $\K$ are also in general
position---no pair of them are tangent and no three have a common
boundary point.\footnote{%
   When we carry over the analysis to the \pdf model, the latter
   assumption will hold with probability $1$ when $\pi$ is Lebesgue-continuous,
   but may fail for a discrete probability mass distribution. In the
   latter situation, we can use symbolic perturbations to turn $\pi$
   into a density in general position, without affecting the asymptotic
   bound that we are after.}
Using the standard symbolic perturbation techniques (see e.g.~\cite[Chapter~7]{SA}), the proof can be
extended when $\E$ or $\RRR$ is not  in general position or when some of
the expansion distances are $0$; we omit here the routine details.

For each
$1 \le i \le n$, let $\lpt_i, \rpt_i$ be the left and right endpoints,
respectively, of $e_i$ (as mentioned, we assume, with no loss of generality, that no
segment in $\E$ is vertical). We draw a random permutation $\sigma$ of
$\{1,\ldots,n\}$, and, for each $1\le i \le n$, we put $r_i =
\theta_{\sigma(i)}$. We then form the Minkowski sums $K_i = e_i \oplus
D(r_i)$, for $i=1,\ldots,n$.  We refer to such a $K_i$ as a
\emph{racetrack}. Its boundary consists of two
semicircles $\gamma_i^-$ and $\gamma_i^+$, centered at the respective endpoints
$\lpt_i$ and $\rpt_i$ of $e_i$, and of two parallel copies, $e_i^-$
and $e_i^+$, of $e_i$; we use $e_i^-$ (resp., $e_i^+$) to denote the
straight edge of $K_i$ lying below (resp., above) $e_i$. Let
$\lpt_i^-,\lpt_i^+$ (resp., $\rpt_i^-,\rpt_i^+$) denote the left
(resp., right) endpoints of $e_i^-, e_i^+$, respectively. We regard
$K_i$ as the union of two disks $D_{i}^-$, $D_{i}^+$ of radius $r_i$
centered at the respective endpoints $\lpt_i$, $\rpt_i$ of $e_i$, and
a rectangle $R_i$ of width $2r_i$ having $e_i$ as a midline. The left
endpoint $\lpt_i$ splits the edge $\lpt_i^-\lpt_i^+$ of $R_i$ into two
segments of equal length, and similarly $\rpt_i$ splits the edge
$\rpt_i^-\rpt_i^+$ of $R_i$ into two segments of equal lengths. We
refer to these four segments $\lpt_i\lpt_i^-, \lpt_i\lpt_i^+,
\rpt_i\rpt_i^-, \rpt_i\rpt_i^+$ as the \emph{\halfedges} of $R_i$;
see \figref{racetrack}(i).
Set $\K=\{K_i\mid 1 \le i \le n\}$,
$\D=\{D_i^+,D_i^- \mid 1 \le i \le n\}$,
and $\R = \{R_i \mid 1 \le i \le n\}$.

As above, let $\union = \union(\K)$ denote the union of $\K$. We show
that the expected number of vertices on $\bd \union$ is $O(n\log n)$,
where the expectation is over the choice of the random permutation
$\sigma$ that produces the distances $r_1,\ldots,r_{n}$.

\begin{figure}[\si{htb}]
    \centering
    \begin{tabular}{\si{ccc}}
%        \scalebox{0.5}
        {
           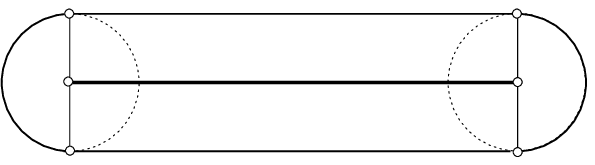%
        }
        &&
        %\scalebox{0.5}
        {%
           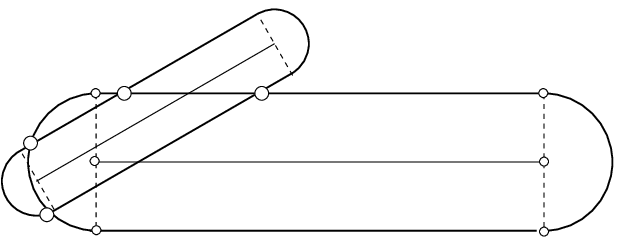%
        }%
        \\
        \small (i)&&\small (ii)
    \end{tabular}
    \caption{XXX (i) Segment $e_i$, racetrack $K_i$, and its constituents
       rectangle $R_i$ and disks $D_i^-,D_i^+$. (ii) Union of two
       racetracks. $\alpha, \beta$ are RR-vertices, $\zeta$ is a
       CR-vertex, and $\xi$ is a CC-vertex; $\alpha$ is a non-terminal
       vertex and $\beta$ is a terminal vertex (because of the edge of
       $R_j$ it lies on, which ends inside $R_i$).}%
    \figlab{racetrack}
\end{figure}

We classify the vertices of $\bd \union$ into three types (see
\figref{racetrack}(ii)):
\begin{itemize}
    \item[(i)] \emph{CC-vertices}, which lie on two semicircular arcs
    of the respective pair of racetrack boundaries;
    \item[(ii)] \emph{RR-vertices}, which lie on two straight-line
    edges; and
    \item[(iii)] \emph{CR-vertices}, which lie on a semicircular arc
    and on a straight-line edge.
\end{itemize}

Bounding the number of CC-vertices is trivial because they are also
vertices of $\union (\D)$, the union of the $2n$ disks $D_i^-$,
$D_i^+$, so their number is $O(n)$~\cite{APS,KLPS}. We therefore focus
on bounding the expected number of RR- and CR-vertices of $\bd\union$.

\subsection{RR-vertices}
Let $v$ be an RR-vertex of $\union$, lying on $\partial R_i$ and
$\partial R_j$, the rectangles of two respective segments $e_i$ and
$e_j$. Denote the edges of $R_i$ and $R_j$ containing $v$ as $\eta_i
\in\{e_i^-,e_i^+\}$ and $\eta_j \in \{e_j^-,e_j^+\}$, respectively.  A
vertex $v$ is \emph{terminal} if either a subsegment of $\eta_i$
connecting $v$ to one of the endpoints of $\eta_i$ is fully contained
in $K_j$, or a subsegment of $\eta_j$ connecting $v$ to one of the
endpoints of $\eta_j$ is fully contained in $K_i$; otherwise $v$ is a
\emph{non-terminal} vertex. For example, in \figref{racetrack}(ii),
$\beta$ is a terminal vertex, and $\alpha$ is a non-terminal vertex.
There are at most $4n$ terminal vertices on $\bd\union$, so it
suffices to bound the expected number of non-terminal vertices.

Our strategy is first to describe a scheme that charges each non-terminal
RR-vertex $v$ to one of the \halfedges of one of the rectangles of $\R$
on whose boundary $v$ lies, and then to
prove that the expected number of vertices charged to each \halfedge
is $O(\log n)$. The bound $O(n\log n)$ on the expected number of
(non-terminal) RR-vertices then follows.
% \micha{Either remove the rest of the para,
   %or add the ``negative-slope'' condition too.}  The only properties
%that we require from the charging is that when a non-terminal
%RR-vertex $v$, lying on $\bd R_i\cap \bd R_j$, for a corresponding
%pair of segments $e_i$, $e_j$, is charged to a \halfedge $g$ then (i)
%$g$ is a \halfedge of $R_i$ or $R_j$, say $R_j$, adjacent to the edge
%of $R_j$ containing $v$, and (ii) the other segment, in this case
%$e_i$, crosses $g$.

%%%%%%%%%%%%%%%%%%%%%%%%%%%%%%%%%%%%%%%%%%%%%%%%%%%%%%%%%%%%%%%%%%%%%

Let $v$ be a non-terminal RR-vertex lying on $\bd R_i\cap\bd R_j$, for
two respective input segments $e_i$ and $e_j$.  To simplify the
notation, we rename $e_i$ and $e_j$ as $e_1$ and $e_2$. Note that $v$
is an intersection of $e^+_1$ or $e_1^-$ with $e_2^+$ or
$e_2^-$. Since the analysis is the same for each of these four choices
we will say that $v$ is the intersection of $e_1^\pm$ with $e_2^\pm$
where $e_1^\pm$ (resp., $e_2^\pm$) is either $e^+_1$ (resp., $e_2^+$)
or $e_1^-$ (resp., $e_2^-$).

For $i=1,2$, let $g_i^\pm$ denote the (unique) portion of $e_i^\pm$
between $v$ and an endpoint $w_i$ so that, locally near $v$, $g_i^\pm$
is contained in the second racetrack $K_{3-i}$ (that is, in
$R_{3-i}$). Now take $Q_i$ to be the rectangle with $g_i^\pm =vw_i$ as
one of its sides, and with the orthogonal projection of $g_i^\pm$ onto
$e_i$ as the opposite parallel side. We denote this edge of $Q_i$,
which is part of $e_i$, by $E_i$. We let $A_i$ be the side
perpendicular to $E_i$ and incident to $w_i$. Note that $A_i$ is one
of the four \halfedges of $R_i$. We denote by $A_i^*$ the side of
$Q_i$ that is parallel to $A_i$ (and incident to $v$). See
Figure~\ref{qandr}.

Each rectangle $Q_i$, $i=1,2$, has a complementary rectangle $Q'_i$ on
the same side of $e_i$, which is openly disjoint from $Q_i$, so that
$Q_i\cup Q'_i$ is the half of the rectangular portion $R_i$, between
$e_i$ and $e_i^\pm$, of the full racetrack $K_i$.

\begin{figure}[\si{htb}]
    \begin{center}
        \includegraphics{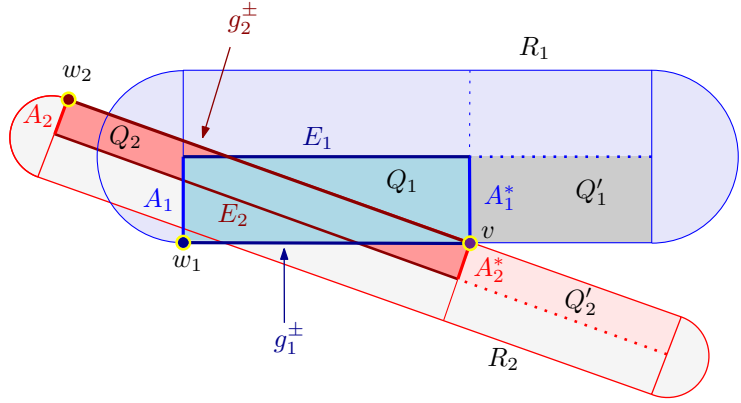}%
        % \scalebox{0.5}{\input{figs/qandr.pstex_t}}
        \caption{The rectangles $Q_1$ and $Q_2$ defined for a
           non-terminal RR-vertex $v$.}%
        \label{qandr}
    \end{center}
\end{figure}

Since $E_i\subset e_i$ for $i=1,2$, it follows that $E_1$ and $E_2$
are disjoint. Since $v$ is a non-terminal RR-vertex, $w_1$ lies
outside $K_2$ and $w_2$ lies outside $K_1$. So, as we walk along the
edge $vw_i$ of $Q_i$ from $v$ to $w_i$, we enter, locally near $v$,
into the other racetrack $K_{3-i}$, but then we have to exit it before
we get to $w_i$. Note that either of these walks, say from $v$ to
$w_1$, may enter $Q_2$ or keep away from $Q_2$ (and enter instead the
complementary rectangle $Q'_2$); see Figures~\ref{qandr} and
\ref{rect2b2c}(i) for the former situation, and Figure~\ref{rect2b2c}(ii)
for the latter one.

Another important property of the rectangles $Q_1$ and $Q_2$, which
follows from their definition, is that they are ``oppositely
oriented'', when viewed from $v$, in the following sense. When we view
from $v$ one of $Q_1$ and $Q_2$, say $Q_1$, and turn counterclockwise,
we first see $E_1$ and then $A_1$, and when we view $Q_2$ and turn
counterclockwise, we first see $A_2$ and then $E_2$.

So far our choice of which among the segments defining $v$ is denoted
by $e_1$ and which is denoted by $e_2$ was arbitrary. But in the rest
of our analysis we will use $e_1$ to denote the segment such that when
we view $Q_1$ from $v$ and turn counterclockwise, we first see $E_1$
and then $A_1$. The other segment is denoted by $e_2$.

The following lemma provides the key ingredient for our charging
scheme.

%%%%%%%%%%%%%%%%%%%%%%%%%%%%%
\begin{lemma}%
    \lemlab{chg}%
    Let $v$ be a non-terminal $RR$-vertex. Then, in the terminology
    defined above, one of the edges, say $E_1$, has to intersect
    either the \halfedge $A_2$ of the other rectangle $Q_2$, or the
    \halfedge $A'_2$ of the complementary rectangle $Q'_2$.
\end{lemma}
%%%%%%%%%%%%%%%%%%%%%%%%%%%%%

\begin{proof}
    For $i=1,2$, we associate with $Q_i$ a \emph{viewing arc}
    $\Gamma_i$, consisting of all orientations of the rays that
    emanate from $v$ and intersect $Q_i$. Each $\Gamma_i$ is a
    quarter-circular arc (of angular span $90^\circ$), which is
    partitioned into two subarcs $\Gamma_i^A$, $\Gamma_i^E$, at the
    orientation at which $v$ sees the opposite vertex of the
    corresponding rectangle $Q_i$; $\Gamma_i^A$ (resp., $\Gamma_i^E$)
    is the subarc in which we view $A_i$ (resp., $E_i$).  See
    Figure~\ref{2arcs2}.

    \begin{figure}[\si{htb}]
        \begin{center}
%            \scalebox{0.5}%
            {%
               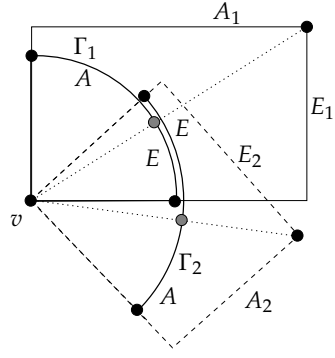%
            }
            \caption{The two viewing arcs $\Gamma_1$, $\Gamma_2$,
               their partitions into $\Gamma_1^E$, $\Gamma_1^A$ and
               $\Gamma_2^A$, $\Gamma_2^E$, and their
               overlap.} \label{2arcs2}
        \end{center}
    \end{figure}

    Moreover, the opposite orientations of $Q_1$ and $Q_2$ mean that,
    as we trace these arcs in counterclockwise direction, $\Gamma_1^E$
    precedes $\Gamma_1^A$, whereas $\Gamma_2^E$ succeeds $\Gamma_2^A$.
    That is, the clockwise endpoint of $\Gamma_1$ is adjacent to
    $\Gamma_1^E$ and we call it the \emph{E-endpoint} of $\Gamma_1$,
    and the counterclockwise endpoint of $\Gamma_1$, called the
    \emph{A-endpoint}, is adjacent to $\Gamma_1^A$. Symmetrically, the
    clockwise endpoint of $\Gamma_2$ is its A-endpoint and is adjacent
    to $\Gamma_2^A$, and its counterclockwise endpoint is the
    E-endpoint, adjacent to $\Gamma_2^E$. See Figures~\ref{2arcs2}
    and~\ref{rect2b2c}.

    \begin{figure}[\si{htb}]
        \centering%
        \begin{tabular}{\si{ccc}}
            %\scalebox{0.55}
            {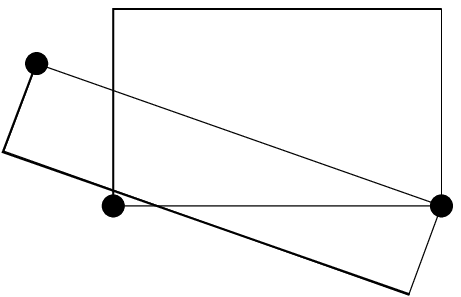}
            &\;\;\;\;\;\;&
            %\scalebox{0.55}
            {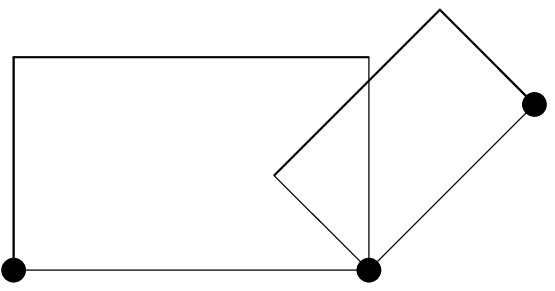}%
            \\
            {\small (i)} && {\small (ii)}
        \end{tabular}
        \begin{center}
            \caption{(i) One possible interaction between $Q_1$ and
               $Q_2$. The overlap is of type $AA$, the intersection
               $\Gamma_0$ of the viewing arcs is delimited by the
               orientations of $\vec{vw}_1$ and $\vec{vw}_2$. (i)
               Another possible interaction between $Q_1$ and
               $Q_2$. The overlap is of type $EE$, the intersection
               $\Gamma_0$ of the viewing arcs does not contain the
               orientations of $\vec{vw}_1$ and
               $\vec{vw}_2$. } \label{rect2b2c}
        \end{center}
    \end{figure}

    Finally, the overlapping of $Q_1$ and $Q_2$ near $v$ mean that the
    arcs $\Gamma_1$ and $\Gamma_2$ overlap too. Let
    $\Gamma_0:=\Gamma_1\cap\Gamma_2$.

    The viewing arcs $\Gamma_1$ and $\Gamma_2$ can overlap in one of
    the following two ways.

    \begin{description}
        \item{{\bf AA-overlap:}} The clockwise and the
        counterclockwise endpoints of $\Gamma_0$ are the A-endpoints
        of $\Gamma_2$ and $\Gamma_1$, respectively. See
        Figure~\ref{rect2b2c}(i).
        \item{{\bf E{}E-overlap:}} The clockwise and the
        counterclockwise endpoints of $\Gamma_0$ are the E-endpoints
        of $\Gamma_1$ and $\Gamma_2$, respectively. See
        Figure~\ref{rect2b2c}(ii).
    \end{description}

    We now assume that none of the four intersections (between one of
    the segments and a suitable \halfedge of the other rectangle),
    mentioned in the statement of the lemma, occur. We reach a
    contradiction by showing that under this assumption neither type
    of overlap can happen.

    \paragraph{AA-overlap.}
    For $i=1,2$, let $\rho_i(\theta)$, for $\theta\in\Gamma_i$, denote
    the length of the intersection of $Q_i$ with the ray emanating
    from $v$ in direction $\theta$. Note that $\rho_i$ is
    \emph{bimodal}: it increases to its maximum, which occurs at the
    direction to the vertex of $Q_i$ opposite to $v$, and then
    decreases back (each of the two pieces is a simple trigonometric
    function of $\theta$).

    Write $\Gamma_0 = [\alpha,\beta]$. Since the overlap of $\Gamma_1$
    and $\Gamma_2$ is an AA-overlap $\alpha$ is the orientation of
    $\vec{vw}_2$, and we have $\rho_2(\alpha) > \rho_1(\alpha)$ (we
    have to exit $Q_1$ before we reach $w_2$).  Symmetrically, we have
    $\rho_2(\beta) < \rho_1(\beta)$. See
    Figure~\ref{rect2b2c}(i). Hence, by continuity, there must exist
    $\theta\in\Gamma_0$ where $\rho_1(\theta)=\rho_2(\theta)$.

    \begin{figure}[\si{htb}]
        \begin{center}
            %\scalebox{0.5}
            {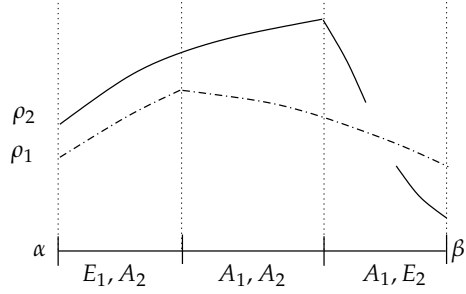}
            \caption{Illustrating the argument that $\rho_1$ and
               $\rho_2$ cannot intersect in an AA-overlap.}
            \label{bimodaa}
        \end{center}
    \end{figure}

    We claim however that this is impossible. Indeed, by taking into
    account the partitions $\Gamma_1=\Gamma_1^E\cup\Gamma_1^A$,
    $\Gamma_2=\Gamma_2^A\cup\Gamma_2^E$, and by overlaying them within
    $\Gamma_0$, we see that $\Gamma_0$ is partitioned into at most
    three subarcs, each being the intersection (within $\Gamma_0$) of
    one of $\Gamma_1^E$, $\Gamma_1^A$ with one of $\Gamma_2^A$,
    $\Gamma_2^E$.  See Figure~\ref{bimodaa} (and also
    Figure~\ref{2arcs2}, where the overlap has only two subarcs, of
    $\Gamma_1^E\cap\Gamma_2^E$ and $\Gamma_1^A\cap\Gamma_2^E$).  Since
    $E_1$ and $E_2$ are disjoint, and since, by assumption, no
    $E$-edge of any rectangle intersects the $A$-edge of the other
    rectangle, the intersection $\rho_1(\theta)=\rho_2(\theta)$ can
    only occur within $\Gamma_1^A\cap\Gamma_2^A$. As we trace
    $\Gamma_0$ from $\alpha$ to $\beta$, we start with $\rho_2 >
    \rho_1$, so this still holds as we reach
    $\Gamma_1^A\cap\Gamma_2^A$. However, the bimodality of $\rho_1$,
    $\rho_2$ and the different orientations of $Q_1$, $Q_2$ mean that
    $\rho_1$ is \emph{decreasing} on $\Gamma_1^A$, whereas $\rho_2$ is
    \emph{increasing} on $\Gamma_2^A$, so no intersection of these
    functions can occur within $\Gamma_1^A\cap\Gamma_2^A$, a
    contradiction that shows that an AA-overlap is impossible.

    \paragraph{E{}E-overlap.}
    We follow the same notations as in the analysis of AA-overlaps,
    but use different arguments, which bring to bear the complementary
    rectangles $Q'_1$, $Q'_2$.

    Consider the clockwise endpoint $\alpha$ of $\Gamma_0$, which, by
    construction, is the E-endpoint of $\Gamma_1$, incident to
    $\Gamma_1^E$. Consider first the subcase where $\rho_1(\alpha) >
    \rho_2(\alpha)$. That is, the edge $A^*_1$ of $Q_1$ connecting $v$
    to $E_1$ crosses and exits $Q_2$ before reaching $E_1$; it may
    exit $Q_2$ either at $E_2$ (as depicted in
    Figure~\ref{rect2b2c}(ii)) or at $A_2$ (as depicted in
    Figure~\ref{rect2d}).

    \begin{figure}[\si{htb}]
        \centering
        \begin{tabular}{\si{ccc}}
            %\scalebox{0.5}
            {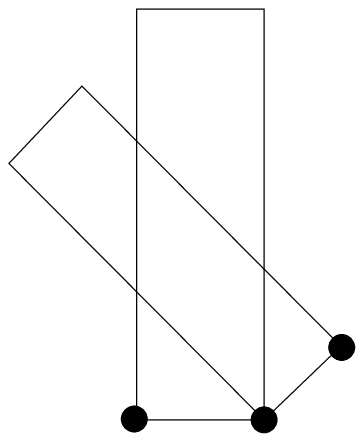}
            &\;\;\;\;\;\;\;\;&
            %\scalebox{0.5}
            {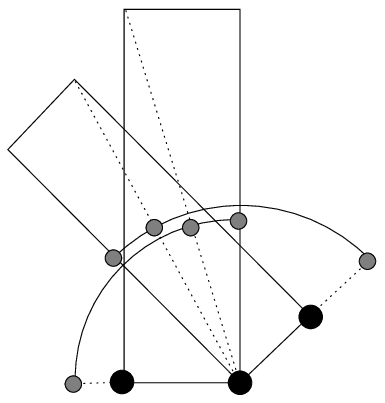}
            \\
            {\small (i)} && {\small (ii)}
        \end{tabular}
        \begin{center}
            \caption{(i) Another instance of an E{}E-overlap. (ii) The
               intersection of the viewing arcs; it consists of three
               subarcs (in counterclockwise order):
               $\Gamma_1^E\cap\Gamma_2^A$, $\Gamma_1^A\cap\Gamma_2^A$,
               and $\Gamma_1^A\cap\Gamma_2^E$. Here $w_2$ lies in the
               diametral disk (not drawn) spanned by the edge $A^*_1$
               of $Q_1$.} \label{rect2d}
        \end{center}
    \end{figure}

    If $A_1^*$ exits $Q_2$ at $E_2$ then we follow $E_2$ into the
    complementary rectangle $Q'_1$. By our assumption (that no
    intersection as stated in the lemma occurs) $E_2$ cannot exit
    $Q'_1$ through its anchor side $A'_1$ (as depicted in
    Figure~\ref{rect2bx}(i)). So $E_2$ must end inside $Q'_1$, at an
    endpoint $q_2$ (see Figure~\ref{rect2bx}(ii)). But then the right
    angle $q_2w_2v$ must either cross the anchor $A'_1$ twice, or be
    fully contained in $Q'_1$. In the latter case $w_2$ lies in
    $Q'_1\subset R_1$, contrary to the assumption that $v$ is
    non-terminal, and in the former case $w_2$ lies in the disk with
    $A'_1$ as a diameter, which is also contained in $K_1$, and again
    we have a contradiction.

    \begin{figure}[\si{htb}]
        \centering
        \begin{tabular}{\si{ccc}}
            %\scalebox{0.5}
            {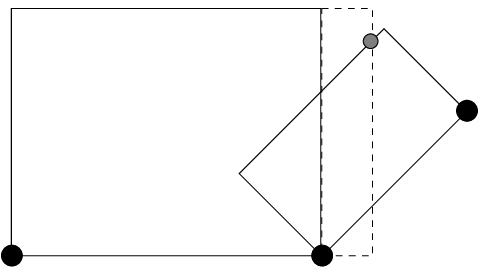}
            &\;\;\;\;\;\;\;\;&
            %\scalebox{0.5}
            {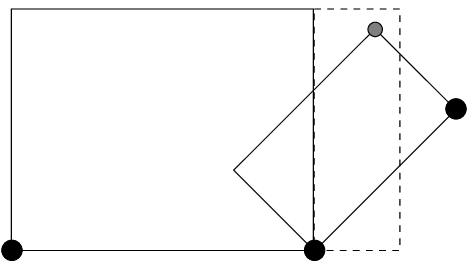}\\
            {\small (i)} && {\small (ii)}
        \end{tabular}
        \begin{center}
            \caption{Another instance of an E{}E-overlap. The
               intersection of the viewing arcs consists of just
               $\Gamma_1^E\cap\Gamma_2^E$. (i) $E_2$ crosses the
               anchor $A'_1$ of the complementary rectangle $Q'_1$.
               (ii) $E_2$ ends inside $Q'_1$.} \label{rect2bx}
        \end{center}
    \end{figure}

    If $A_1^*$ exits $Q_2$ at $A_2$, the argument is simpler, because
    then $w_2$ is contained in the disk with $A_1^*$ as a diameter,
    which is contained in $K_1$, again contrary to the assumption that
    $v$ is non-terminal (see Figure~\ref{rect2d}).

    A fully symmetric argument leads to a contradiction in the case
    where $\rho_1(\beta) < \rho_2(\beta)$. It therefore remains to
    consider the case where $\rho_1(\alpha) < \rho_2(\alpha)$ and
    $\rho_1(\beta) > \rho_2(\beta)$. Here we argue exactly as in the
    case of AA-overlaps, using the bimodality of $\rho_1$ and
    $\rho_2$, that this case cannot happen. (Figure~\ref{bimodaa}
    depicts the situation in this case too.) Specifically, there has
    to exist an intersection point of $\rho_1$ and $\rho_2$ within
    $\Gamma_0$, and it can only occur at
    $\Gamma_1^A\cap\Gamma_2^A$. But over this subarc $\rho_1$ is
    decreasing and $\rho_2$ is increasing, and we enter this subarc
    with $\rho_1 < \rho_2$, so these functions cannot intersect within
    this arc. This completes the argument showing that our assumption
    implies that an E{}E-overlap is not possible.

    We conclude that one of the intersections stated in the lemma must
    exist.
\end{proof}

\paragraph{The charging scheme.}
We charge $v$ to a \halfedge ($A_2$ or $A_2'$) of $R_2$ that
intersects $E_1$ or to a \halfedge ($A_1$ or $A_1'$) of $R_1$ that
intersects $E_2$. At least one such intersection must exist by
Lemma~\ref{lem:chg}. A useful property of this charging, which will be
needed in the next part of the analysis, is given by the following
lemma.
%%%%%%%%%%%%%%%%%%%%%%%
\begin{lemma}%
    \label{lem:neg}%
    Let $v$ be a non-terminal RR-vertex, lying on $\bd R_i\cap \bd
    R_j$, which is charged to a \halfedge $h_j$ of $R_j$. Then $e_i$,
    traced from its intersection with $h_j$ into $R_j$, gets further
    away from $e_j$.
\end{lemma}
%%%%%%%%%%%%%%%%%%%%%%%

\begin{figure}[\si{htb}]
    \begin{center}
        %\scalebox{0.5}
        {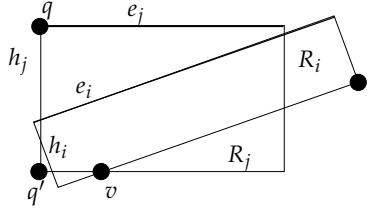}
        \caption{Illustrating the proof of Lemma~\ref{lem:neg}. Only
           the lower portions of $R_i$ and $R_j$ are shown.}
        \label{nonpos}
    \end{center}
\end{figure}

\begin{proof}
    Suppose to the contrary that $e_i$ approaches $e_j$ and assume,
    without loss of generality, that $e_j$ is horizontal, that $v$
    lies on $e_j^-$, and that $h_j$ is the left-lower \halfedge of
    $R_j$. In this case $e_i$ has positive slope. See
    Figure~\ref{nonpos}.

    Let $q$ denote the endpoint of $h_j$ incident to $e_j$ and let
    $q'$ denote the lower endpoint of $h_j$. (Note that $q'$ is $w_j$
    if $v$ is charged to the \halfedge $A_j$ of $R_j$ which is also a
    \halfedge of $Q_j$ and $q'$ is the endpoint of $A_j'$, the
    \halfedge of the complementary rectangle $Q_j'$, otherwise.)  Since
    $v$ is a non-terminal RR-vertex, the segment $\vec{vq'}$, as we
    trace it from $v$, enters $R_i$ (this follows as $e_i$ has
    positive slope) and then exits it before reaching $q'$.  The exit
    point lies on a suitable \halfedge $h_i$ of $R_i$. Since $e_i$
    intersects $h_j$, it follows that $h_i$ must also cross $h_j$.
    But then $q'$ must lie inside the diametral disk spanned by $h_i$,
    and thus it lies inside $K_i$, a contradiction that completes the
    proof.
\end{proof}

%%%%%%%%%%%%%%%%%%%%%%%%%%%%%%%%%%%%%%%%%%%%%%%%%%%%%

% First, I think that the dependence of the ``charging'' on the radii
% is not a real issue.  We have three kinds of random variables. The
% first is just the number of RR-vertices, and the second ones are the
% number of RR-vertices charged to each anchor.  The first is the sum
% of the second ones. Now each anchor also has a third variable, which
% is the number of intervals of the union on the shifted edge (such as
% $e_j^-$). The sum of the third type of variables is an upper bound
% for the other two types, and linearity of expectation does the
% rest. Am I missing something here? I hope not.

%%%%%%%%%%%%%%%%%%%%%%%%%%%%%%%%%%%%%%%%%%%%%%%%%%%%%%%%%%%%%%%%%%%%%%%%%

\paragraph{The expected number of vertices charged to a \halfedge.}
Fix a segment of $\E$, denote it as $e_0$, and rename the other
segments as $e_1,\dots,e_{n-1}$.  Assume, for simplicity, that $e_0$
is horizontal. We bound the expected number of vertices charged to the
lower-left \halfedge, denoted by $g$, of the rectangle $R_0$ (which is
incident to the left endpoint, $\lpt_0$, of $e_0$); symmetric
arguments will apply to the other three \halfedges of $R_0$.  Given a
specific permutation $(r_0,\ldots,r_{n-1})$ of the input set of
distances $\RRR$, let $\chi_\RR(g;r_0,\ldots,r_{n-1})$ denote the
number of vertices charged to $g$ if $e_i$ is expanded by $r_i$, for
$i=0,\ldots,n-1$.  We wish to bound $\bar\chi_\RR(g)$, the expected
value of $\chi_\RR(g;r_0,\ldots,r_{n-1})$ with respect to the random
choice of the $r_i$'s, as effected by randomly shuffling them (by a
random permutation acting on $\RRR$).

We first fix a value $r$ (one of the values $\theta_i\in\RRR$)
of $r_0$ and bound $\bar \chi_\RR(g\mid r)$,
the expected number of vertices charged to $g$ conditioned on the
choice $r_0=r$; the expectation is taken over those permutations that
fix $r_0=r$; they can be regarded as \emph{random} permutations of the
remaining elements of $\RRR$.  Then we bound $\bar\chi_\RR(g)$ by
averaging the resulting bound over the choice of $r_0$.

So fix $r_0=r$. Set $K_0 = e_0\oplus D(r)$, and let $\ell_0^-$ denote
the line supporting $e_0^-$. We have $g=\lpt_0\lpt_0^-$, and observe
that all these quantities depend only on $r_0$, so they are now
fixed. By our charging scheme, if a vertex $v \in \bd R_0 \cap \bd
R_j$ is charged to the \halfedge $g$, then $v \in e_0^-$, and $e_j$
intersects $g$. Furthermore, by Lemma~\ref{lem:neg}, the slope of
$e_j$ is negative. Let $\E_g\subseteq \E \setminus \{e_0\}$ be the set
of segments that intersect $g$ and have negative slopes; the set
$\E_g$ depends on the choice of $r_0=r$ but not on (the shuffle of)
$r_1, \ldots, r_{n-1}$.

For a fixed permutation $(r_1,\ldots,r_{n-1})$, set $\K_g = \{ K_l := e_l\oplus D(r_l)
\mid e_l \in \E_g\}$ and $\union_g = \union(\K_g) \cap e_0^-$. We call
a vertex of $\union_g$ an \emph{R-vertex} if it lies on $\bd R_i$ for
some $e_i \in \E_g$ (as opposed to lying on some semicircular arc). If
a non-terminal RR-vertex $v$ is charged to the \halfedge $g$, then $v$
is an R-vertex of $\union_g$ (for the specific choice $r_0=r$). It
thus suffices to bound the expected number of R-vertices on
$\union_g$, where the expectation is taken over the random shuffles of
$r_1, \ldots, r_{n-1}$.

% \paragraph{Expected number of R-vertices in $\union_\dela$}
Consider a segment $e_i \in \E_g$. If $\ell_0^-\cap e_i \not=
\emptyset$ then we put $q_i = \ell_0^-\cap e_i$. If $\ell_0^-\cap e_i
= \emptyset$, then let $\lambda_i$ denote the line perpendicular to
$e_i$ through $\rpt_i$ (the right endpoint of $e_i$), and define $q_i$
to be the intersection of $\lambda_i$ with $\ell_0^-$. (We may assume
that $q_i$ lies to the right of $\lpt_0^-$, for otherwise no expansion
of $e_i$ will be such that $R_i$ intersects the edge $e_0^-$.) Define
$r_i^* = 0$ if $\ell_0^-\cap e_i \not= \emptyset$ and $r_i^*=|\rpt_i
q_i|$ otherwise. For simplicity, write $\E_g$ as $\langle e_1, \ldots,
e_m\rangle$, for some $m < n$, ordered so that $q_1,\ldots,q_m$ appear on
$\ell_0^-$ from \emph{right to left} in this order; see
\figref{vwzprob}. We remark that $q_1,\ldots,q_m$ are independent of
the values of $r_1,\ldots,r_m$, and that the order $e_1,\ldots,e_m$
may be different from the order of the intercepts of these segments
along $g$ (e.g., see the segments $e_1$ and $e_2$ in \figref{vwzprob}).

\begin{figure}[\si{htb}]
    \centering%
    {\includegraphics[width=0.5\linewidth]{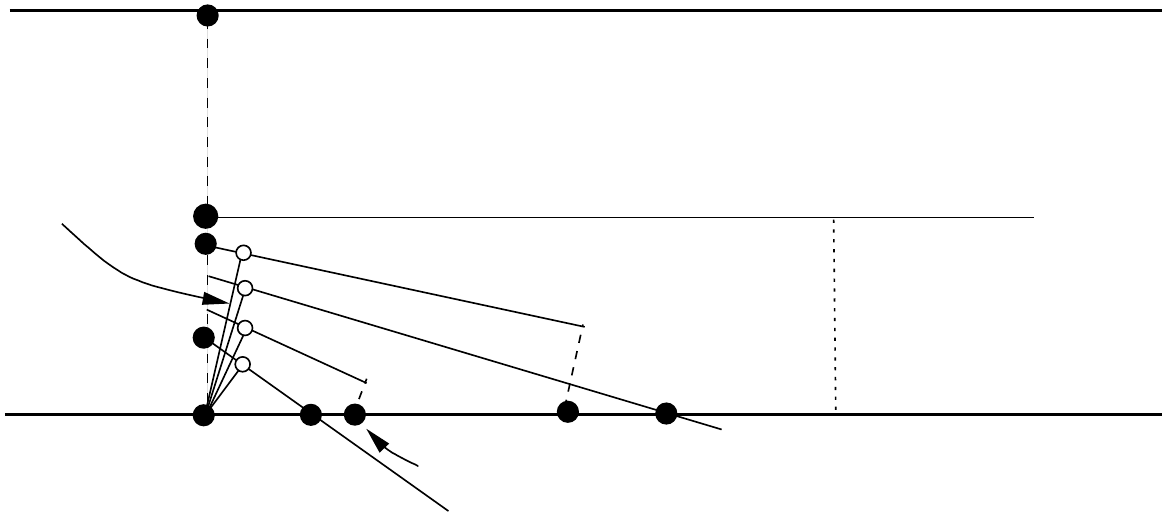}}

    \caption{Segments in $\E_g$ and the points that they induce on
       $\ell_0^-$.}
    \figlab{vwzprob}
\end{figure}

For $i=1,\ldots,m$, let $r_i$ be, as above, the
(random) expansion distance chosen for $e_i$, and set
$J_i = R_i \cap \ell_0^-$. If $r_i \le r_i^*$ then $J_i=\emptyset$,
and if $r_i > r_i^*$ then $J_i$ is an interval containing $q_i$.
Let $\union_0$ be the union of the intervals $J_i$, and let
$\mu (r; r_1,\ldots r_m)$ be the number of connected components of
$\union_0$. %It can be verified that
Clearly, each R-vertex of $\union_g$ is an endpoint of a component of
$\union_0$, which implies that $\chi_\RR(g; r, r_1,\ldots,r_{n-1}) \le
2\mu(r; r_1,\ldots,r_m)$. It therefore suffices to bound $\bar\mu(r)$, the
expected value of $\mu(r; r_1,\ldots,r_m)$ over the random shuffles of
$r_1,\ldots,r_m$.

For each $e_i \in \E_g$, let $\beta_i$ be the length of the segment
connecting $\lpt_0^-$ to its orthogonal projection on $e_i$. As is
easily checked, we have $\beta_i < r$. It is also clear that if $r_i
\ge \beta_i$ then the entire segment $q_i \lpt_0^-$ is contained in
$K_i$.

%%%%%%%%%%%%%%%%%%%%%%%%%%%%%%%%%%%%%%%%%
\begin{lemma} \label{RRprob}
    In the preceding notations, the expected
    value of $\bar\mu(r)$ is $O(\log n)$.
\end{lemma}

%%%%%%%%%%%%%%%%%%%%%%%%%%%%%%%%%%%%%%%%%
\begin{proof}
    Assume that $r=\radii_{n-k+1}$, for some $k\in\{1,\ldots,n\}$.  We
    claim that is this case $\bar\mu(r) \le n/(k+1)$.

    For $i=1,\ldots,m$, if $r_i > r$ then $r_i > \beta_i$ and
    therefore $\lpt_0^- \in J_i$. Hence, if $i$ is the smallest index
    for which $r_i > r$ (assuming that such an index exists), then
    $\union_0$ has at most $i$ connected components: the one
    containing $J_i$ and at most $i-1$ intervals to its right.

    Recall that we condition the analysis on the choice of $r_0=r$,
    and that we are currently assuming that $r_0$ is the $k$-th
    largest value of $\RRR$.  For this fixed value of $r_0$, the set
    $\E_g$ is fixed.

    Order the segments in $\E_0:=\E\setminus e_0$ by placing first the
    $m$ segments of $\E_g$ in their order as defined above, and then
    place the remaining $n-m-1$ segments in an arbitrary
    order. Clearly this reshuffling of the segments does not affect
    the property that the expansion distances in
    $\RRR_0:=\RRR\setminus\{r\}$ that are assigned to them form a
    random permutation of $\RRR_0$.

    In this context, $\bar{\mu}(r)$ is upper bounded by the expected
    value of the index $j$ of the first segment $e_j$ in $\E_0$ that
    gets one of the $k-1$ distances larger than $r$. (In general, the
    two quantities are not equal, because we set
    $\mu(r; r_1,\ldots,r_m)=m$ when $j$ is greater than $m$, that is,
    in case no segment of $\E_g$ gets a larger distance.)

    As is well known, the expected value of $j$ is $n/k$ (this follows,
    e.g., as in \cite[p.~175,~Problem~2]{GKP}), from which our claim follows.
    (Note that the case $k=1$ is special, because no index can get a larger
    value, but the resulting expectation, namely $n$, serves as an upper
    bound for $\bar{\mu}(r)$.)

    Since $r=\radii_{n-k+1}$ with probability $1/n$, for every $k$, we have
    $$
    \EE[\bar\mu(r)] = \sum_{k=1}^n \frac1n \cdot
    \bar{\mu}(\radii_{n-k+1}) \le \frac1n \sum_{k=1}^n \frac{n}{k} =
    \sum_{k=1}^n \frac{1}{k} = O(\log n) .
    $$
\end{proof}

\paragraph{Putting it all together.}
Lemma~\ref{RRprob} proves that the expected number of non-terminal
vertices of $\union$ charged to a fixed \halfedge of some rectangle in
$\R$ is $O(\log n)$. By Lemma~\ref{lem:chg}, each non-terminal RR-vertex of
$\union$ is charged to one of the $4n$ \halfedges of the rectangles in
$\R$. Repeating this analysis for all these $4n$ \halfedges, the
expected number of non-terminal RR-vertices in $\union$ is $O(n\log
n)$. Adding the linear bound on the number of terminal RR-vertices, we
obtain the following result.

\begin{lemma} \label{lem:rr} The expected number of RR-vertices of
    $\union(\K)$ is $O(n\log n)$.
\end{lemma}

%%%%%%%%%%%%%%%%%%%%%%%%%%%%%%%%%%%%%%%%%%%%%%%%%%%%%%%%%%%%%%%
\subsection{CR-vertices}
%%%%%%%%%%%%%%%%%%%%%%%%%%%%%%%%%%%%%%%%%%%%%%%%%%%%%%%%%%%%%%%
%
Next, we bound the expected number of CR-vertices of $\union$. Using a
standard notation, we call a vertex $v \in \union$ lying on $\bd K_i
\cap \bd K_j$ \emph{regular} if $\bd K_i$ and $\bd K_j$ intersect at
two points (one of which is $v$); otherwise $v$ is called
\emph{irregular}. By a result of Pach and Sharir~\cite{PS}, the number
of regular vertices on $\bd \union$ is proportional to $n$ plus the
number of irregular vertices on $\bd\union$. Since the expected number
of RR- and CC-vertices on $\bd \union$ is $O(n\log n)$, the number of
regular CR-vertices on $\bd \union$ is $O(n\log n+\kappa)$, where
$\kappa$ is the number of irregular CR-vertices on $\bd\union$. It
thus suffices to prove that $\kappa = O(n\log n)$.

\paragraph{Geometric properties of CR-vertices.}
We begin by establishing a few simple geometric lemmas.

% \begin{lemma}
%     Let $D \in \D$ be a disk. If a segment $e_i \in \E$ is disjoint
%     from $D$, then $\bd K_i$ intersects $D$ in at most two points.
% \end{lemma}

\begin{figure}[t]
    \centering%
    %\scalebox{0.5}
    {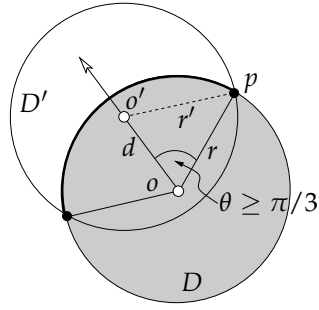}
    \caption{Illustration of the proof of
       Lemma~\protect\ref{lem:disks}.} \label{fig:disks}
\end{figure}

\begin{lemma}
    \label{lem:disks} Let $D$ and $D'$ be two disks of respective
    radii $r,r'$ and centers $o,o'$. Assume that $r'\ge r$ and that
    $o'\in D$.  Then $D'\cap \bd D$ is an arc of angular extent at
    least $2\pi/3$, centered at the radius vector of $D$ from $o$
    through $o'$.
\end{lemma}

\begin{proof}
    We may assume that $D$ is not fully contained in $D'$, for
    otherwise the claim is trivial. Consider then the triangle $oo'p$,
    where $p$ is one of the intersection points of $\bd D$ and $\bd
    D'$.  Put $|oo'|=d\le r$, and let $\angle o'op=\theta$; see
    Figure~\ref{fig:disks}. Then
    $$
    \cos\theta = \frac{r^2+d^2-r'^2}{2dr} \le \frac{d^2}{2dr} =
    \frac{d}{2r} \le \frac12 .
    $$
    Hence $\theta \ge \pi/3$. Since the angular extent of $D'\cap \bd
    D$ is $2\theta$, the claim follows. The property concerning the
    center of the arc $D'\cap \bd D$ is also obvious.
\end{proof}

\begin{cor}
    \label{cor:disks} Let $D$ and $D'$ be two disks of radii $r$ and
    $r'$ and centers $o$ and $o'$, respectively, let $D_1$ be a sector
    of $D$ of angle $\pi/3$, and let $\gamma_1$ denote the circular
    portion of $\bd D_1$. (a) If $o'\in D_1$ and $r < r'$ then
    $\gamma_1$ is fully contained in $D'$. (b) If $o'\notin D_1$ then
    either $D'$ is disjoint from $\gamma_1$ or $D'\cap \gamma_1$
    consists of one or two arcs, each containing an endpoint of
    $\gamma_1$.
\end{cor}

\begin{proof}
    The first claim (a) follows from the preceding lemma, since $D'\cap\bd
    D$ is an arc of angular extent at least $2\pi/3$ centered at a
    point on $\gamma_1$. For (b), $D'\cap\bd D$ is a connected arc
    $\delta$, whose center lies in direction $\vec{oo'}$ and thus
    outside $\gamma_1$, and $D'\cap\gamma_1 = \delta\cap\gamma_1$. The
    intersection of two arcs of the same circle consists of zero, one,
    or two connected subarcs.  In the first case the claim is
    obvious. In the third case, each of the arcs $\delta$, $\gamma_1$
    must contain both endpoints of the other arc, so (b) follows.  In
    the second case, the only situation that we need to rule out is
    when $\delta\cap\gamma_1$ is contained in the relative interior of
    $\gamma_1$, so $\delta$, and its center, are contained in
    $\gamma_1$, contrary to assumption. Hence (b) holds in this case
    too.
\end{proof}

Fix a segment of $\E$, call it $e_0$, and rename the other segments to
be $e_1,\dots,e_{n-1}$. $\partial K_0$ has two semicircular arcs, each
corresponding to a different endpoint of $e_0$.  We fix one of the
semicircular arcs of $K_0$ and denote it by $\gamma_0$. Let $r_0$ be
the random distance assigned to $e_0$, let $D_0$ be the disk of radius
$r_0$ containing $\gamma_0$ on its boundary, and let $H_0\subset D_0$
be the half-disk spanned by $\gamma_0$.

Partition $H_0$ into three sectors of angular extent $\pi/3$ each,
denoted as $H_{01},H_{02},H_{03}$. Let $\gamma_{0i} \subset\gamma_0$
denote the arc bounding $H_{0i}$, for $i=1,2,3$. Here we call a vertex
$v \in \bd \union$ formed by $\gamma_{0i} \cap \bd K_j$, for some $j$,
a \emph{terminal} vertex if $K_j$ contains one of the endpoints of
$\gamma_{0i}$, and a \emph{non-terminal} vertex otherwise.  There are
at most six terminal vertices on $\gamma_0$, for an overall bound of
$12n$ on the number of such vertices, so it suffices to bound the
(expected) number of non-terminal irregular CR-vertices on each subarc
$\gamma_{0i}$, for $i=1,2,3$.

Let $\E(r_0)$ denote the set of all segments $e_j\ne e_0$ that
intersect the disk $D_0$, and, for $i=1,2,3$, let $\E_i(r_0)\subseteq
\E(r_0)$ denote the set of all segments $e_j\ne e_0$ that intersect
the sector $H_{0i}$. Set $m_i := m_i(r_0) = |\E_i(r_0)|$. Segments in
$\E(r_0)\setminus \E_i(r_0)$ intersect $D_0$ but are disjoint from
$H_{0i}$.  (The parameter $r_0$ is to remind us that all these sets
depend (only) on the choice of $r_0$.)

\begin{lemma}
    \label{lem:boundary}
    Let $e_j\in \E \setminus \E(r_0)$.  If $\union$ has a CR-vertex $v
    \in \gamma_{0i}\cap \bd K_j$, for some $i =1,2,3$, then $v$ is
    either a regular vertex or a terminal vertex.
\end{lemma}
\begin{proof}
    Let $c$ denote the center of $D_0$, and consider the interaction
    between $K_j$ and $D_0$. We split into the following two cases.

    \noindent {\bf Case 1, $r_j \le r_0$:} Regard $D_0$ as
    $D_0^*\oplus D(r_j)$, where $D_0^*$ is the disk of radius
    $r_0-r_j$ centered at $c$. By assumption, $D_0^*$ and $e_j$ are
    disjoint, implying that $D_0$ and $K_j$ are pseudo-disks
    (cf.~\cite{KLPS}), that is, their boundaries intersect in two
    points, one of which is $v$; denote the other point as $v'$.

    If only $v$ lies on $\gamma_0$, then $v$ must be a terminal
    vertex, so assume that both $v$ and $v'$ lie on $\gamma_0$. We
    claim that $\bd K_j$ and $\bd K_0$ can intersect only at $v$ and
    $v'$, implying that $v$ is regular. Indeed, $v$ and $v'$ partition
    $\bd K_j$ into two connected pieces. One piece is inside $D_0$,
    locally near $v$ and $v'$, and cannot intersect $\bd K_i$ in a
    point other than $v$ and $v'$ without intersecting $D_0$ in a
    third point (other than $v$ and $v'$), contradicting that $D_0$
    and $K_j$ are pseudo-disks. The other connected piece of $\bd K_j$
    between $v$ and $v'$ is separated from $\bd K_i \setminus
    \gamma_0$ by the line through $v$ and $v'$ and therefore cannot
    contain intersections other than $v$ and $v'$ between $\bd K_j$
    and $\bd K_i$. See Figure \ref{fig:cr}(a).

    \noindent {\bf Case 2, $r_j > r_0$:} Let $K_j^* = e_j \oplus
    D(r_j-r_0)$. $K_j$ can now be regarded as $K_j^* \oplus
    D(r_0)$. If $c \not\in K_j^*$, then by the result of \cite{KLPS},
    $D_0=c\oplus D(r_0)$ and $K_j$ are pseudo-disks; see
    Figure~\ref{fig:cr}(b).  Therefore, the argument given above for
    the case where $r_0\ge r_j$ implies the lemma in this case as
    well. Finally, $c \in K_j^*$ implies that $K_j$ contains $D_0$, so
    this case cannot occur (it contradicts the existence of $v$). See
    Figure~\ref{fig:cr}(c).
\end{proof}

\begin{figure*}[t]
    \centering
    \begin{tabular}{cc}
        %\scalebox{0.5}
        {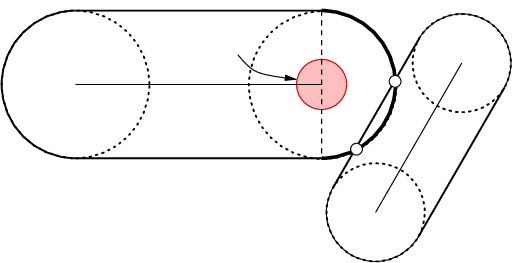}&
        %\scalebox{0.5}
        {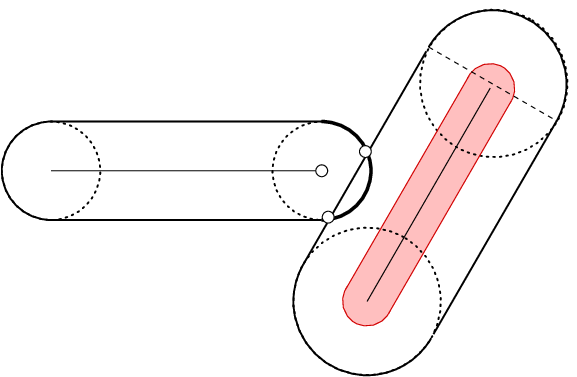} \\
        \small (a)&\small(b) \\
        \multicolumn{2}{c}{
           %\scalebox{0.5}
           {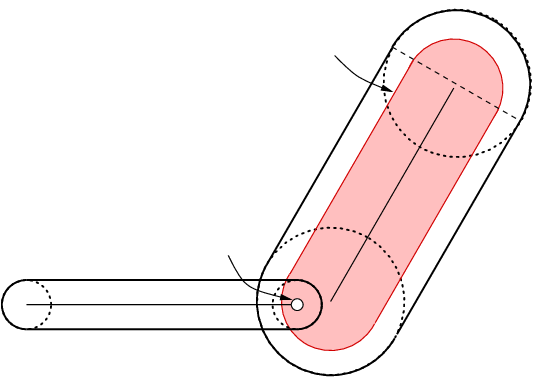}}\\
        \multicolumn{2}{c}{\small (c)}\\
    \end{tabular}
    \caption{(a): The case when $r_0 \ge r_j$ and $\bd K_j \cap
       \gamma_0$ contains the two intersection points of $\bd D_0$ and
       $\bd K_j$. (b) The case when $r_j > r_0$ and $c \not\in
       K_j^*$. (c) The case when $r_j > r_0$ and $c \in K_j^*$.}%
    \label{fig:cr}
\end{figure*}

Using Lemmas~\ref{lem:disks} and~\ref{lem:boundary}, we obtain the
following property.

%%%%%%%%%%%%%%%%%%%%%%%%%%%%%%%%%%%%%%%%%
\begin{lemma}
    \label{lem:non-terminal} Let $v\in \gamma_{0i}\cap \bd K_j$ be a
    non-terminal, irregular CR-vertex of $\union$. Then (i) $e_j \in
    \E_i(r_0)$, and (ii) for all $e_l \in \E_i(r_0)$, $r_l < r_0$.
\end{lemma}
%%%%%%%%%%%%%%%%%%%%%%%%%%%%%%%%%%%%%%%%%
\begin{proof}
    Lemma~\ref{lem:boundary} implies that $e_j \in \E(r_0)$. Suppose
    first that $e_j \in \E(r_0)\setminus \E_i(r_0)$. Pick a point
    $o'\in e_j\cap D_0$, which exists by assumption, and note that
    $K_j$ contains the disk $D'$ of radius $r_j$ centered at
    $o'$. Part (b) of Corollary~\ref{cor:disks} implies that $D'$
    intersects $\gamma_{0i}$ at an arc or a pair of arcs, each
    containing an endpoint of $\gamma_{0i}$, i.e., $v$ is a terminal
    vertex, contrary to assumption.  We can therefore conclude that
    $e_j \in \E_i(r_0)$. Part (a) of Corollary~\ref{cor:disks} implies
    that $r_l < r_0$ for all $e_l \in \E_i(r_0)$, because otherwise we
    would have $\gamma_{0i}\subset K_l$ and $\gamma_{0i}$ would not
    contain any vertex of $\bd\union$.
\end{proof}

\paragraph{The expected number of non-terminal vertices on
   $\gamma_0$.}
We are now ready to bound the expected number of non-terminal
irregular CR-vertices of $\union$ that lie on the semi-circular arc
$\gamma_0$ of $K_0$. Note that $\gamma_0$ is not fixed, as it depends
on the value of $r_0$. Let $\chi_\CR(\gamma_0; r_0, r_1, \ldots,
r_{n-1})$ denote the number of non-terminal irregular vertices on
$\gamma_0$, assuming that $r_i$ is the expansion distance of $K_i$,
for $i=0,\ldots,n-1$. Our goal is to bound
\[ \bar\chi_\CR(\gamma_0) = \EE[\chi_\CR(\gamma_0;
r_0,\ldots,r_{n-1})] \] where the expectation is over all the random
permutations assigning these distances to the segments of $\E$.  As
for RR-vertices, we first fix the value of $r_0$ to, say, $r$, and
bound $\chi_\CR(\gamma_0\mid r)$, the expected value of
$\chi_\CR(\gamma_0, r, r_1, \ldots, r_{n-1})$, where the expectation
is taken over the random shuffles of $r_1, \ldots, r_{n-1}$, and then
bound $\bar\chi_\CR(\gamma_0)$ by averaging over the choice of $r_0$.

%%%%%%%%%%%%%%%%%%%%%%%%%%%%%%%%%%%%
\begin{lemma}
    Using the notation above, $\bar\chi_\CR(\gamma_0) = O(\log n)$.
\end{lemma}
%%%%%%%%%%%%%%%%%%%%%%%%%%%%%%%%%%%%
\begin{proof}
    Following the above scheme, suppose that the value $r_0$ is indeed
    fixed to $r$, so $\gamma_0$ and $\gamma_{0i}$, $1 \le i \le 3$,
    are fixed. As above, set $m_i= |\E_i(r)|$, for $i=1,2,3$; the sets
    $\E_i(r)$ and their sizes $m_i$ are also fixed.  We bound the
    expected number of non-terminal irregular vertices on
    $\gamma_{0i}$, for a fixed $i\in\{1,2,3\}$. By
    Lemma~\ref{lem:non-terminal}, any such vertex lies on the boundary
    of $\J_0$, the intersection of $\gamma_{0i}$ with the union of
    $\{K_l \mid e_l \in \E_i(r)\}$. Equivalently, it suffices to
    bound the expected number of connected components of $\J_0$ that
    lie in the interior of $\gamma_{0i}$. By
    Lemma~\ref{lem:non-terminal}, if $r_l \ge r$ for any $e_l \in
    \E_i(r)$, then there are no such components.

    Assume that $r=\radii_{n-k+1}$ for some $k\in\{1,\ldots,n\}$.  To
    bound $\chi_\CR(\gamma_0\mid r)$, we first bound the probability
    $p$ that all the $m_i$ radii that are assigned to the segments of
    $\E_i(r)$ are smaller than $r$. We have
    \begin{align*}
        p & =%
        \frac{ \binom{n-k}{m_i} }{ \binom{n-1}{m_i} } = \frac{
           (n-k)(n-k-1)\cdots (n-k-m_i+1) }
        { (n-1)(n-2)\cdots (n-m_i) }  \\
        & =%
        \left(1-\frac{k-1}{n-1}\right) \left(1-\frac{k-1}{n-2}\right)
        \cdots
        \left(1-\frac{k-1}{n-m_i}\right) \\
        & < \left(1-\frac{k-1}{n-1}\right)^{m_i} < e^{-(k-1)m_i/(n-1)} .
    \end{align*}
    That is, with probability $1-p$ there are no connected components.
    Note that $1-p=0$ when $k=1$.
    In the complementary case, when all the $m_i$ radii under
    consideration are smaller than $r$, we pessimistically bound the
    number of connected components by $2m_i$ --- each segment of
    $\E_i(r)$ can generate at most two connected components. In other
    words, when $k\ge 2$, the expected number of connected components
    of $\J_0$ is at most
    $$
    2m_ip < 2m_i e^{-(k-1)m_i/(n-1)} = \frac{2(n-1)}{k-1} \cdot \left( ((k-1)m_i/(n-1))
        e^{-(k-1)m_i/(n-1)}\right) < \frac{2(n-1)}{e(k-1)} ,
    $$
    because the maximum value of the expression $xe^{-x}$ is $e^{-1}$.
    The bound is $2m_i \le 2(n-1)$ when $k=1$.

    Since $r=\radii_{n-k+1}$ with probability $1/n$ for every $k$, we have
    \begin{align*}
    \EE[\bar\chi_\CR(\gamma_0) ]%
    &=%
    \EE[\chi_\CR(\gamma_0\mid r)] = \sum_{k=1}^n \frac1n \cdot
    \EE[\chi_\CR(\gamma_0\mid \radii_{n-k+1})] \\
    &\le \frac1n \left[ 2(n-1) + \sum_{k=2}^n
    \frac{2(n-1)}{e(k-1)} \right] \\
    &= O\left( \sum_{k=1}^n \frac{1}{k} \right) = O(\log n) .
    \end{align*}
\end{proof}

Summing this bound over all three subarcs of $\gamma_0$ and adding the
constant bound on the number of terminal (irregular) vertices, we
obtain that the expected number of irregular CR-vertices of $\union$
on $\gamma_0$ is $O(\log n)$.  Summing these expectations over the
$2n$ semicircular arcs of the racetracks in $\K$, and adding the
bounds on the number of regular CR-vertices we obtain the following
lemma.
%%%%%%%%%%%%%%%%%%%%%%%%%%%%%%
\begin{lemma} \label{lem:CR} The expected number of CR-vertices on
    $\union(\K)$ is $O(n\log n)$.
\end{lemma}
%%%%%%%%%%%%%%%%%%%%%%%%%%%%%%

Combining Lemma~\ref{lem:rr}, Lemma~\ref{lem:CR}, and the linear bound
on the number of CC-vertices, completes the proof of
Theorem~\ref{thm:main} for the case of segments.

%%%%%%%%%%%%%%%%%%%%%%%%%%%%%%%%%%%%%%%%%%%%%%%%%%%%%%%%%%
\section{The Case of Polygons}
\label{sec:poly}
%%%%%%%%%%%%%%%%%%%%%%%%%%%%%%%%%%%%%%%%%%%%%%%%%%%%%%%%%%

In this section we consider the case where the objects of $\C$ are $n$
convex polygons, each with at most $s$ vertices. For simplicity, we prove
Theorem~\ref{thm:main} when each $C_i$ is a convex $s$-gon---if $C_i$ has
fewer than $s$ vertices, we can split some of its edges into multiple edges so that it has exactly $s$ vertices.
We reduce this case to the case of segments
treated above. A straightforward reduction that just takes the edges
of the $s$-gons as our set of segments does not work since edges of
the same polygon are all expanded by the same distance.
Nevertheless, we can overcome this difficulty as follows.

Let $\C=\{C_1,\ldots,C_n\}$ be the $n$ given polygons, and consider a
fixed assignment of expansion distances $r_i$ to the polygons $C_i$.
For each $i$, enumerate the edges of $C_i$ as
$e_{i1},e_{i2},\ldots,e_{is}$; the order of enumeration is not
important. Let $v$ be a vertex of $\union$, lying on the boundaries of
$K_i$ and $K_j$, for some $1\le i<j\le n$. Then there exist an edge
$e_{ip}$ of $C_i$ and an edge $e_{jq}$ of $C_j$ such that $v$ lies on
$\bd (e_{ip}\oplus D(r_i))$ and on $\bd (e_{jq}\oplus D(r_j))$; the
choice of $e_{ip}$ is unique if the portion of $\bd K_i$ containing
$v$ is a straight edge, and, when that portion is a circular arc, any
of the two edges incident to the center of the corresponding disk can
be taken to be $e_{ip}$.  A similar property holds for $e_{jq}$.

The following stronger property holds too. For each $1\le p\le s$, let
$\C_p$ be the set of edges $\{e_{1p},e_{2p},\ldots,e_{np}\}$, and let
$\K_p = \{ e_{1p} \oplus D(r_1),\ldots,e_{np} \oplus D(r_n) \}$. Then,
as is easily verified, our vertex $v$ is a vertex of the union
$\union(\K_p \cup \K_q)$.  Moreover, for each $p$, the expansion
distances $r_i$ of the edges $e_{ip}$ of $\C_p$ are all the elements
of $\RRR$, each appearing once, and their assignment to the segments
of $\C_p$ is a \emph{random permutation}.
Fix a pair of indices $1\le p<q\le s$, and note that
each expansion distance $r_i$ is assigned to exactly two segments of
$\C_p\cup\C_q$, namely, to $e_{ip}$ and $e_{iq}$.

We now repeat the analysis given in the preceding section for the
collection $\C_p\cup\C_q$, and make the following observations.
First, the analysis of CC-vertices remains the same, since the
complexity of the union of any family of disks is linear.

Second, in the analysis of RR- and CR-vertices, the exploitation of
the random nature of the distances $r_i$ comes into play only after we
have fixed one segment (that we call $e_0$) and its expansion distance
$r_0$, and consider the expected number of RR-vertices and CR-vertices
on the boundary of $K_0=e_0\oplus D(r_0)$, conditioned on the fixed
choice of $r_0$. Suppose, without loss of generality, that $e_0$
belongs to $\C_p$. We first ignore its sibling $e'_0$ in $\C_q$ (from
the same polygon), which receives the same expansion distance $r_0$;
$e'_0$ can form only $O(1)$ vertices of $\union$ with
$e_0$.\footnote{%
   As a matter of fact, $e_0$ and $e'_0$ do not generate any vertex of
   the full union $\union(\C)$, but they might generate vertices of
   $\union(\C_p\cup\C_q)$.} The interaction of $e_0$ with the other
segments of $\C_p$ behaves exactly as in Section~\ref{sec:segs}, and
yields an expected number of $O(\log n)$ RR-vertices of $\union(\K_p)$
charged to the \halfedges of $R_0$ and an expected number of $O(\log
n)$ CR-vertices charged to circular arcs of $K_0$.  Similarly, The
interaction of $e_0$ with the other segments of $\C_q$ (excluding
$e_0'$) is also identical to that in Section~\ref{sec:segs}, and
yields an additional expected number of $O(\log n)$ vertices of
$\union(\{e_0\}\cup\K_q)$ charged to an \halfedges and circular arcs
of $K_0$. Since any vertex of $\union(\K_p\cup\K_q)$ involving $e_0$
must be one of these two kinds of vertices, we obtain a bound of
$O(\log n)$ on the expected number of such vertices, and summing this
bound over all segments $e_0$ of $\C_p\cup\C_q$, we conclude that the
expected complexity of $\union(\K_p\cup\K_q)$ is $O(n\log n)$. (Note
also that the analysis just given manages to finesse the issue of
segments sharing endpoints.)

Summing this bound over all $O(s^2)$ choices of $p$ and $q$, we obtain
the bound asserted in Theorem~\ref{thm:main}. The constant of
proportionality in the bound that this analysis yields is $O(s^2)$.

%%%%%%%%%%%%%%%%%%%%%%%%%%%%%%%%%%%%%%%%%%%%%%%%%%%%%%%%%%%%%%%%%
%%%%%%%%%%%%%%%%%%%%%%%%%%%%%%%%%%%%%%%%%%%%%%%%%%%%%%%%%%%%%%%%%
\section{Network Vulnerability Analysis}
\label{sec:nva}
%%%%%%%%%%%%%%%%%%%%%%%%%%%%%%%%%%%%%%%%%%%%%%%%%%%%%%%%%%%%%%%%%
%%%%%%%%%%%%%%%%%%%%%%%%%%%%%%%%%%%%%%%%%%%%%%%%%%%%%%%%%%%%%%%%%

\begin{figure}%[\si{htb}]
    \centering
    \includegraphics[width=2.5in]{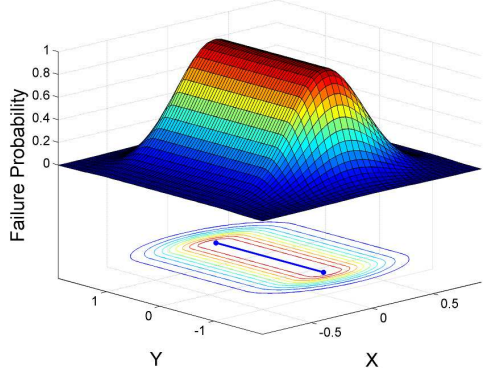}
    \caption{Discretizing the function $f_e$ for an edge $e$.}
    \label{fig:link}
\end{figure}

Let $\E=\{e_1, \ldots, e_n\}$ be a set of $n$ segments in the plane
with pairwise-disjoint relative interiors,
   and let $\ph:\preals\rightarrow [0,1]$ be an
edge failure probability function such that $1-\ph$ is a \cdf. For
each segment $e_i$, define the function $f_i: \reals^2 \rightarrow
[0,1]$ by $f_i(q)=\ph(d(q,e_i))$, for $q\in\reals^2$, where
$d(q,e_i)$ is the distance from $q$ to $e_i$, and set
$$\Phi(q,\E) = \sum_{i=1}^n f_i(q).$$
In this section we present a Monte-Carlo algorithm, which is an
adaptation and a simplification of the algorithm described
in~\cite{AEG*}, for computing a location $\tilde q$ such that
$\Phi(\tilde q,\E) \ge (1-\delta)\Phi(\E)$, where $0<\delta<1$ is some
prespecified error parameter, and where
$$\Phi(\E) = \max_{q\in\reals^2} \Phi(q,\E).$$
The expected running time of the algorithm is
a considerable improvement over the algorithm in \cite{AEG*};
this improvement is a consequence of the bounds obtained in the
preceding sections.

To obtain the algorithm we first discretize each $f_i$ by choosing a
finite family $\K_i$ of \emph{super-level sets} of $f_i$ (each of the
form $\{q\in\reals^2 \mid f_i(q) \ge t\}$), and reduce the problem of
computing $\Phi(\E)$ to that of computing the maximum depth in the
arrangement $\A(\K)$ of $\K=\bigcup_i \K_i$. Our algorithm then uses a
sampling-based method for estimating the maximum depth in $\A(\K)$,
and thereby avoids the need to construct $\A(\K)$ explicitly.
% The efficiency of our algorithm is established using Corollary~\ref{cor:klevel}.

In more detail, set $m=\lceil 2n/\delta\rceil$.
For each $1 \le j < m$, let $r_j = \ph^{-1}(1-j/m)$, and let,
for $i=1,\ldots,n$,
$K_{ij}=e_i \oplus D(r_j)$ be the racetrack formed by the Minkowski
sum of $e_i$ with the disk of radius $r_j$ centered at the origin.
Note that $r_j$ increases with $j$. Set $\tilde\ph=\{ r_j \mid 1 \le j
< m\}$, $\K_i=\{ K_{ij} \mid 1 \le j < m\}$, and $\K=\bigcup_{1 \le i
   \le n} \K_i$.  See Figure~\ref{fig:link}.
Note that we cannot afford to, and indeed do not, compute $\K$
explicitly, as its cardinality (which is quadratic in $n$) is too
large.

For a point $q \in \reals^2$ and for a subset $\X\subseteq \K$, let
$\Delta(q,\X)$, the \emph{depth} of $q$ with respect to $\X$, be the
number of racetracks of $\X$ that contain $q$ in their interior, and
let
$$\Delta(\X) = \max_{q \in \reals^2} \Delta(q,\X).$$
The following lemma (whose proof, which is straightforward, can be
found in~\cite{AEG*}) shows that the maximum depth of $\K$
approximates $\Phi(\E)$

%%%%%%%%%%%%%%%%%%%%%%%%%%%%%%%%%%%%%%%%%%%%%%%%%%%%%%%%%%%%%%
\begin{lemma}[Agarwal~\etal~\cite{AEG*}]
    \label{lem:discrete} (i) ${\displaystyle \Phi(q,\E) \ge
       \frac{\Delta(q,\K)}{m} \ge \Phi(q,\E) - \tfrac12\delta}$ for
    each point $q\in\reals^2$.\\ (ii) $\displaystyle \Phi(\E) \ge
    \frac{\Delta(\K)}{m} \ge (1-\tfrac12\delta) \Phi(\E)$.
\end{lemma}
%%%%%%%%%%%%%%%%%%%%%%%%%%%%%%%%%%%%%%%%%%%%%%%%%%%%%%%%%%%%%%

% \noindent (Briefly, it follows from the definitions that, for any
% $q\in \reals^2$, $f_i(q)$ is well approximated by $1/m$ times the
% number of racetracks of $\K_i$ that contain $q$, with an additive
% error of at most $1/m$; specifically, the latter quantity lies
% between $f_i(q)-\tfrac{1}{m}$ and $f_i(q)$. Summing over $i$, we
% obtain that $\Phi(q,\E) \ge \tfrac{\Delta(q,\K)}{m} \ge
% \Phi(q,\E)-\tfrac{n}{m}$; since $n/m \le \delta/2$, (i) follows.
% Since we always have $\Phi(\E) \ge 1$, we get a relative error of at
% most $(\delta/2) \Phi(\E)$ at the point $q^*$ that attains
% $\Phi(\E)$, and part (ii) then easily follows.)

By Lemma~\ref{lem:discrete}, it suffices to compute a point $\tilde q$
of depth at least $(1-\tfrac12\delta)\Delta(\K)$ in $\A(\K)$; by (i) and
(ii) we will then have
\[
\Phi(\tilde q,\E) \ge \frac{\Delta(\tilde q,\K)}{m} \ge
(1-\tfrac12\delta)\frac{\Delta(\K)}{m} \ge (1-\tfrac12\delta)^2
\Phi(\E) > (1-\delta) \Phi(\E) .
\]

We describe a Monte-Carlo algorithm for computing
such a point $\tilde q$, which is
a simpler variant of the algorithm described in~\cite{AHP} (see
also~\cite{AHSSW}), but we first need the following
definitions.

For a point $q \in \reals^2$ and for a subset $\X\subseteq \K$, let
$\fdepth(q,\X) = \tfrac{\Delta (q,\X)}{|\X|}$ be the \emph{fractional
   depth} of $q$ with respect to $\X$, and let
$$\fdepth(\X) = \max_{q\in\reals^2}\ \fdepth(q,\X)
= \frac{\Delta(\X)}{|\X|} \; .$$ We observe that $\Delta(\K) \ge m-1$
because the depth near each $e_i$ is at least $m-1$. Hence,
\begin{equation}
    \label{eq:depth}
    \fdepth(\K) \ge \frac{m-1}{|\K|} = \frac{m-1}{(m-1)n} = \frac{1}{n} .
\end{equation}

Our algorithm estimates fractional depths of samples of $\K$ and
computes a point $\tilde q$ such that $\fdepth({\tilde q},\K) \ge
(1-\tfrac12\delta)\fdepth(\K)$. By definition, this is equivalent to
$\Delta({\tilde q},\K) \ge (1-\tfrac12\delta)\Delta(\K)$, which is
what we need.

%%%%%%%%%%%%%%%%%%%%%%%%%%%%%%%%%%%%%%%%%%%%%%%%%%%%%%%%%
\medskip

We also need the following concept from the theory of random sampling.

For two parameters $0 < \rho,\eps < 1$, we call a subset $\AA\subseteq
\K$ a \emph{$(\rho,\eps)$-approximation} if the following holds for
all $q \in \reals^2$:
\begin{equation}
    \label{eq:approx}
    \left| \fdepth(q,\K) - \fdepth (q,\AA) \right| \le
    \begin{cases}
        \eps\fdepth(q,\K) & \text{if $\fdepth (q,\K) \ge \rho$} \\
        \eps\rho & \text{if $\fdepth (q,\K) < \rho$} .
    \end{cases}
\end{equation}
This notion of $(\rho,\eps)$-approximation is a special case of the
notion of \emph{relative $(\rho,\eps)$-approximation} defined
in~\cite{HS11} for general range spaces with finite VC-dimension.  The
special case at hand applies to the so-called dual range space
$(\K,\reals^2)$, where the ground set $\K$ is our collection of
racetracks, and where each point $q\in\reals^2$ defines a range equal
to the set of racetracks containing $q$; here $\Delta(q,\K)$ is the
size of the range defined by $q$, and $\fdepth(q,\K)$ is its
\emph{relative size}.  Since $(\K,\reals^2)$ has finite VC-dimension
(see, e.g., \cite{SA}), it follows from a result in~\cite{HS11} that,
for any integer $b$, a random subset of size
\begin{equation}
    \asize(\rho,\eps) := \frac{cb}{\eps^2\rho} \ln n
\end{equation}
is a $(\rho,\eps)$-approximation of $\K$ with probability at least
$1-1/n^b$, where $c$ is a sufficiently large constant (proportional to
the VC-dimension of our range space). In what follows we fix $b$ to be
a sufficiently large integer, so as to guarantee (via the probability
union bound) that, with high probability, all the samplings that we
construct in the algorithm will have the desired approximation
property.

The algorithm works in two phases. The first phase finds a value $\rho
\ge 1/n$ such that $\fdepth(\K) \in [\rho, 2\rho]$.  The second phase
exploits this ``localization'' of $\fdepth(\K)$ to compute the desired
point $\tilde q$.

The first phase performs a decreasing exponential search: For $i \ge
1$, the $i$-th step of the search tests whether $\fdepth(\K) \le
1/2^i$.  If the answer is \textsc{yes}, the algorithm moves to the
$(i+1)$-st step; otherwise it switches to the second phase. Since we
always have $\fdepth(\K) \ge 1/n$ (see (\ref{eq:depth})), the first
phase consists of at most $\lceil \log_2 n\rceil$ steps.

At the $i$-th step of the first phase, we fix the parameters
$\rho_i=1/2^i$ and $\eps=1/8$, and construct a
$(2\rho_i,\eps)$-approximation of $\K$ by choosing a random subset
$\R_i \subset \K$ of size $\asize_i=\asize (2\rho_i,\eps)=O(2^i\log
n)$. We construct $\A(\R_i)$, e.g., using the randomized incremental
algorithm described in~\cite[Chapter~4]{SA} and compute
$\fdepth(\R_i)$ by traversing the arrangement $\A(\R_i)$.  Then, if
$$\fdepth(\R_i)\le (1-2\eps)\rho_i = \tfrac34 \rho_i \ ,$$
we continue to step $i+1$ of the first phase.  Otherwise, we switch to
the second phase of the algorithm (which is described below). The
following lemma establishes the important properties of the first
phase.
\begin{lemma} \label{lem:navi} When the algorithm reaches step $i$ of
    the first phase, we have $\fdepth(\K) \le \rho_{i-1}$ and
    $\fdepth(\K)$ is within the interval
    $[\fdepth(\R_i)-\tfrac{1}{4}\rho_i,
    \fdepth(\R_i)+\tfrac{1}{4}\rho_i]$.
\end{lemma}
\begin{proof}
    The proof is by induction on the steps of the algorithm. Assume
    that the algorithm is in step $i$. Then by induction $\fdepth(\K)
    \le \rho_{i-1}= 2\rho_i$ (for $i=1$ this is trivial since
    $\rho_{i-1} =1$). This, together with $\R_i$ being a
    $(2\rho_i,\eps)$-approximation of $\K$, implies by
    \eqref{eq:approx} that
    \begin{equation} \label{eq:111} \fdepth(\K) = \fdepth(q^*,\K) \le
        \fdepth(q^*,\R_i) + 2\eps\rho_i \le \fdepth(\R_i) +
        2\eps\rho_i = \fdepth(\R_i) + \tfrac14\rho_i ,
    \end{equation}
    where $q^*$ is a point satisfying $\fdepth(q^*,\K)=\fdepth(\K)$.
    Furthermore, using \eqref{eq:approx} again (in the opposite
    direction), we conclude that
    \[ \fdepth(\K) \ge \fdepth(q_i,\K) \ge \fdepth(q_i,\R_i) -
    2\eps\rho_i = \fdepth(\R_i) - \tfrac14\rho_i ,\] where $q_i$ is a
    point satisfying $\fdepth(q_i,\R_i)=\fdepth(\R_i)$. So we conclude
    that $\fdepth(\K)$ is in the interval specified by the lemma.

    The algorithm continues to step $i+1$ if $\fdepth(\R_i) \le
    \tfrac34 \rho_i$. But then by \eqref{eq:111} we get that
    $\fdepth(\K) \le \tfrac34 \rho_i + \frac{1}{4}\rho_i = \rho_i $ as
    required.
\end{proof}

%%%%%%%%%%%%%%%%%%%%%%%%%%%%%%%%%%%%%%%%%%%%%%%%%%%%%%%%

Suppose that the algorithm decides to terminate the first phase and
continue to the second phase, at step $i$. Then, by Lemma 4.2, we have
that $\fdepth(\K)\in [\fdepth(\R_i)-\tfrac{1}{4}\rho_i,
\fdepth(\R_i)+\tfrac{1}{4}\rho_i]$. Since, by construction,
$\fdepth(\R_i) > \tfrac34 \rho_i $, the ratio between the endpoints of
this interval is at most $2$, as is easily checked, so if we set
$\rho=\fdepth(\R_i)-\tfrac{1}{4}\rho_i$ then $\fdepth(\K) \in
[\rho,2\rho]$ as required upon entering the second phase.

% Then by the definition of the algorithm we know that $\fdepth(\R_i)
% > \tfrac34 \rho_i $ and by Lemma \ref{lem:navi} we know that
% $\fdepth(\K)\in [\fdepth(\R_i)-\tfrac{1}{4}\rho_i,
% \fdepth(\R_i)+\tfrac{1}{4}\rho_i]$. Since $\fdepth(\R_i) > \tfrac34
% \rho_i $ the ratio between the endpoints of this interval is at most
% $2$, as is easily checked, so if we set
% $\rho=\fdepth(\R_i)-\tfrac{1}{4}\rho_i$ then $\fdepth(\K) \in
% [\rho,2\rho]$ as required upon entering the second phase.

\medskip

In the second phase, we set $\rho=\fdepth(\R_i)-\tfrac{1}{4}\rho_i$
and $\eps=\delta/4$, and construct a $(\rho,\eps)$-approximation of
$\K$ by choosing, as above, a random subset $\R$ of size
$\nu=\asize(\rho,\eps)=O(2^i\log n)$. We compute $\A(\R)$, using the
randomized incremental algorithm in \cite{SA}, and return a point
$\tilde q \in \reals^2$ of maximum depth in $\A(\R)$.

This completes the description of the algorithm.

\paragraph{Correctness.}
We claim that $\omega(\tilde q, \K) \ge (1-\tfrac12\delta)\omega(\K)$.
Indeed, let $q^* \in \reals^2$ be, as above, a point of maximum depth
in $\A(\K)$.  We apply \eqref{eq:approx}, use the fact that
$\fdepth(q^*, \K) \ge \rho$, and consider two cases. If
$\fdepth(\tilde q, \K) \ge \rho$ then
\begin{align*}
    \fdepth(\tilde q, \K) \ge & \frac{\fdepth(\tilde
       q,\R)}{1+\tfrac14\delta} \ge \frac{\fdepth(q^*,\R)}{1+\tfrac14\delta}
    \ge \frac{1-\tfrac14\delta}{1+\tfrac14\delta}\fdepth(q^*,\K)\\
    \ge & (1-\tfrac12\delta)\fdepth( q^*,\K) = (1-\tfrac12\delta)\fdepth(\K) .
\end{align*}
On the other hand, if $\fdepth(\tilde q, \K) < \rho$ then
\begin{align*}
    \fdepth(\tilde q, \K) \ge & \fdepth(\tilde q,\R)
    -\tfrac{\delta}{4} \rho
    \ge \fdepth(q^*,\R) - \tfrac{\delta}{4} \rho \\
    \ge & (1-\tfrac{\delta}{4})\fdepth(q^*,\K) - \tfrac{\delta}{4} \rho \\
    \ge & (1-\tfrac{\delta}{4})\fdepth(q^*,\K)
    - \tfrac{\delta}{4} \fdepth(q^*,\K) \\
    = & (1-\tfrac12\delta)\fdepth(\K) .
\end{align*}
Hence in both cases the claim holds.  As argued earlier, this implies
the desired property
$$\Phi(\tilde q,\E) \ge (1-\delta)\Phi(\E) .$$

\paragraph{Running time.}
We now analyze the expected running time of the algorithm.  We first
note that we do not have to compute the set $\K$ explicitly to obtain
a random sample of $\K$. Indeed, a random racetrack can be chosen by
first randomly choosing a segment $e_i \in \E$, and then by choosing
(independently) a random racetrack of $\K_i$.  Hence, each sample
$\R_i$ can be constructed in $O(\asize_i)$ time, and the final sample
$\R$ in $O(\asize)$ time.

To analyze the expected time taken by the $i$-th step of the first
phase, we bound the expected number of vertices in $\A(\R_i)$.
\begin{lemma}
    The expected number of vertices in the arrangement $\A(\R_i)$ is
    $O(2^i\log^3 n)$.
\end{lemma}
\begin{proof}
    By Lemma \ref{lem:navi} if we perform the $i$-th step of the first
    phase then $\fdepth(\K)\le \rho_{i-1}=2\rho_i$. Therefore, using
    \eqref{eq:approx} we have
    \[ \fdepth(\R_i) \le \fdepth(\K) + 2\eps\rho_i \le 2\rho_i +
    2\eps\rho_i < 3 \rho_i .\] Therefore, $\Delta(\R_i) =
    \fdepth(\R_i)|\R_i| \le 3\rho_i\nu_i = O(\log n)$. The elements in
    $\R_i$ are chosen from $\K$ using the 2-stage random sampling
    mechanism described above, which we can rearrange so that we first
    choose a random sample $\E_i$ of segments, and then, with this
    choice fixed, we choose the random expansion distances. This
    allows us to view $\R_i$ as a set of racetracks over a fixed set
    $\E_i$ of segments, each of which is the Minkowski sum of a
    segment of $\E_i$ with a disk of a random radius, where the radii
    are drawn uniformly at random and independently from the set
    $\tilde\ph$. There is a minor technical issue: we
    might choose in $\E_i$ the same segment $e\in \E$ several times,
    and these copies of $e$ are not pairwise-disjoint. To address this
    issue, we slightly shift these multiple copies of $e$ so as to
    make them pairwise-disjoint.  Assuming that $\E$ is in general
    position and that the \cdf defining $\ph$ is in ``general
    position'' with respect to the locations of the segments of $\E$,
    as defined in Section \ref{sec:segs}, this will not affect the
    asymptotic maximum depth in the arrangement of the sample.
    %A related issue is that the segments of $\E$ may share endpoints. This can also be handled by a small perturbation of these segments.
    % Finally, the analysis in Section~\ref{sec:segs} assumes
    % $\ph$ to be a continuous \pdf, but here $\tilde\ph$ is a
    % discrete distribution. It can be verified that
    % Corollary~\ref{cor:klevel} holds for discrete distributions as
    % well.)

    % \micha{Add (here or elsewhere) a remark about the issue of
     % discretizing $\ph$.}

    By Corollary~\ref{cor:klevel}, applied under the \pdf model and
    conditioned on a fixed choice of $\E_i$, the expected value of
    $|\A(\R_i)|$ is
    $$
    \EE[ |\A(\R_i)| ] = O(\Delta(\R_i)\asize_i\log n)=O(2^i\log^3 n) ,
    $$
    implying the same bound for the unconditional expectation too.
\end{proof}

The expected time spent in constructing $\A(\R_i)$ by the randomized
incremental algorithm in \cite{SA} is $O(\asize_i\log\asize_i +
|\A(\R_i)|)=O(2^i\log^3n)$. Hence, the $i$-th step of the first phase
takes $O(2^i \log^3 n)$ expected time.
% As already remarked, the first phase consists of at most
% $\lceil\log_2 n\rceil$ steps.
Summing this bound over the steps of the first phase, we conclude that
the expected time spent in the first phase is $O(n\log^3 n)$.

In the second phase, $|\R| = O(\tfrac{1}{\delta^{2}\rho}\log n) =
O(\tfrac{n}{\delta^{2}}\log n) ,$ and the same argument as above,
using (\ref{eq:approx}), implies that
$$
\fdepth(\R) \le \max\{ (1+\tfrac{\delta}{4})\fdepth(\K),
\,\fdepth(\K)+ \tfrac{\delta}{4}\rho\} = O(\rho) ;
$$
where the latter bound follows as $\fdepth(\K)\in
[\rho,2\rho]$. Hence, $\Delta(\R)=\fdepth(\R)\cdot|\R| =
O(\tfrac{1}{\delta^2}\log n)$, and the expected size of $\A(\R)$ is
thus $O(\Delta(\R)\cdot|\R|\log n)=O(\tfrac{n}{\delta^4}\log^3 n)$.
Since this dominates the cost of the other steps in this phase, the
second phase takes $O(\tfrac{n}{\delta^4}\log^3 n)$ expected time.

Putting everything together, we obtain that the expected running time
of the procedure is $O(\tfrac{n}{\delta^4}\log^3 n)$, and it computes,
with high probability, a point $\tilde q$ such that $\Phi(\tilde
q,\E)\ge (1-\delta)\Phi(\E)$. This completes the proof of
Theorem~\ref{thm:nva}.

%%%%%%%%%%%%%%%%%%%%%%%%%%%%%%%%%%%%%%%%%%%%%%%%%%%%%%%%%%%%%%%

\section{Discussion}

We have shown that if we take the  Minkowski sums of the members of a family  of
pairwise-disjoint convex sets, each of constant description complexity, with
disks whose radii are chosen using  a suitable  probabilistic model, then
the expected complexity of the union of the Minkowski sums is near linear.
This generalizes the result of Kedem~\etal~\cite{KLPS} and shows that the
complexity of the union of Minkowski sums is quadratic only if the expansion
distances are chosen in an adversial manner. Our model is related to the
so-called \emph{realistic input models}, proposed  to obtain more refined bounds
on the performance of a variety of geometric  algorithms~\cite{BKSV}. There are
also some similarities between our model and the framework of smoothed analysis~\cite{ST}.

A natural collection of open problems is to tighten the bounds in our theorems or
prove corresponding lower bounds. In particular, the following questions arise.
(i) The $O(n^{1+\varepsilon})$ bound of Theorem \ref{thm:sets} is unlikely to
be tight. Is it possible to prove an $O(n\log n)$ upper bound as we
did for polygons in Theorem~\ref{thm:main}?
(ii) Can the bound in Theorem~\ref{thm:main} be improved from $O(s^2n\log n)$ to $O(sn\log n)$?
(iii) Is the bound of Theorem~\ref{thm:main}
asymptotically tight, even for segments, or could one prove a tighter
$o(n\log n)$ bound? maybe even linear?

Another interesting direction for future research is to explore other
problems that can benefit from our model.  For example, we believe that the
expected complexity of the multiplicatively-weighted Voronoi diagram of a set
of points in $\reals^2$ is near-linear if the weights are chosen using one of
our models, and we plan to investigate this problem. Recall that if the weights
are chosen by an adversary, then the complexity is quadratic~\cite{AK}.

Finally, it would be useful to prove, or disprove, that the density and
permutation models are equivalent, in the sense that the value of $\compl(\C)$ is
asymptotically the same under both models for any family $\C$ of pairwise-disjoint convex sets.
Nevertheless, it is conceivable that there is a large class of density functions for which the
\pdf model yields a better upper bound.

\paragraph{Acknowledgments.}
The authors thank Emo Welzl for useful discussions
concerning the two probabilistic models used in the paper.

\end{document}